\documentclass[11pt,a4paper]{article}

\title{Hardness and  Approximation Schemes for Discrete Packing and Domination\thanks{A preliminary version of this paper appeared in the 12th Annual International Conference on Combinatorial Optimization and Applications (COCOA) 2018 \cite{Madireddy2018}.}}
\author{Raghunath Reddy Madireddy$^1$ \and  Apurva Mudgal$^2$ \and Supantha Pandit$^3$}

\date{
	$^1$Birla Institute of Technology and Science, Pilani, Hyderabad Campus,  India \\ \texttt{raghunath@hyderabad.bits-pilani.ac.in}\\%
	$^2$Indian Institute of Technology Ropar, Punjab, India. \\ \texttt{apurva@iitrpr.ac.in}\\
    $^3$Dhirubhai Ambani Institute of Information and Communication Technology, Gandhinagar, Gujarat, India. \\ \texttt{pantha.pandit@gmail.com}\\[2ex]%
}

\usepackage{comment}

\usepackage{proofread}

\usepackage{graphicx}
\usepackage{amsmath,amsthm,amssymb,amsfonts}
\usepackage{epstopdf}

\usepackage{subfigure}
\usepackage{mathtools}
\usepackage{wrapfig}
\usepackage{lipsum}
\usepackage{mathrsfs}
\usepackage[T1]{fontenc}
\usepackage{pifont}

\usepackage[noend,linesnumbered,ruled,lined]{algorithm2e}

\newcommand{\colb}[1]{{\color{black}{{\textit{#1}}}}}

\newcommand{\colr}[1]{{\color{red}{{#1}}}}

\newcommand\norm[1]{\left\lVert#1\right\rVert}

\usepackage{pbox}
\usepackage{array}
\newcolumntype{P}[1]{>{\centering\arraybackslash}p{#1}}
\newcolumntype{M}[1]{>{\centering\arraybackslash}m{#1}}

\newtheorem{claim1}{Claim}

\newcommand{\np}{$\mathsf{NP}$}
\newcommand{\apx}{$\mathsf{APX}$}
\newcommand{\ptas}{$\mathsf{PTAS}$}
\newcommand{\qptas}{$\mathsf{QPTAS}$}

\newcommand{\seg}{\mathsf{seg}}

\usepackage[framemethod=tikz]{mdframed}
\usepackage{lipsum}% just to generate text for the example

\usepackage[framemethod=tikz]{mdframed}
\newmdenv[
  backgroundcolor=gray!15,linewidth=1pt,
  topline=false,
  bottomline=false,
  rightline=false,
  skipabove=\topsep,
  skipbelow=\topsep,
]{siderules}

\definecolor{ccol}{RGB}{152,0,152}

\usepackage{hyperref}

\newcommand{\mds}{Minimum Dominating Set}
\newcommand{\mis}{Maximum Independent Set}
\newcommand{\mdis}{\textit{IS}}
\newcommand{\mdds}{\textit{DS}}

\newcommand{\spds}{$\mathsf{SPECIAL}$-$\mathsf{3DS}$}
\newcommand{\spsc}{$\mathsf{SPECIAL}$-$\mathsf{3SC}$}

\newcommand{\misp}{\textit{MISP-3}}
\newcommand{\mdsp}{\textit{MDSP-3}}

\usepackage{lineno}

\newenvironment{proof*}[1]
  {\renewcommand{\proofname}{{\textbf{#1}}}%
   \begin{proof}}
  {\end{proof}}

\makeatletter
\def\namedlabel#1#2{\begingroup
   \def\@currentlabel{#2}%
   \label{#1}\endgroup
}
\makeatother

 \usepackage[margin=1in]{geometry} % full-width

\newtheorem{theorem}{Theorem}
 
\newtheorem{lemma}{Lemma}
\newtheorem{definition}{Definition}
\begin{document}
	\maketitle
	
	\begin{abstract}
	We present polynomial-time approximation schemes based on \colb{local search} technique for both \colb{geometric (discrete) independent set (\mdis)} and  \colb{geometric (discrete) dominating set (\mdds)} problems, where the objects are arbitrary radii disks and arbitrary side length axis-parallel squares. Further, we show that the \mdds~problem is \apx-hard for various shapes in the plane. Finally, we prove that both \mdis~and \mdds~problems are \np-hard for unit disks intersecting a horizontal line and axis-parallel unit squares intersecting a straight line with slope $-1$. \\

	\end{abstract}

\section{Introduction} \label{np_hard}
 The Maximum Independent Set and the Minimum Dominating Set problems attract researchers due to their numerous applications in various fields of computer science like VLSI design, network routing, etc. The input to both problems consists of a set of geometric objects $\mathcal{R}$ in the plane.  In the \mis~problem, we need to find a maximum size sub-collection of objects $\mathcal{R}^\prime \subseteq \mathcal{R}$ such that no two objects in $\mathcal{R}^\prime$ intersect. In the \mds~problem, we need to find a minimum size sub-collection of objects $\mathcal{R}^\prime \subseteq \mathcal{R}$ such that for every object $O \in (\mathcal{R} \setminus \mathcal{R}^\prime)$ there exists at least one object $O^\prime \in \mathcal{R}^\prime$ such that $O$ and $O^\prime$ intersect. 

%In this paper, we consider \colb{discrete} versions of  both \mis~and \mds~problems, called  the \colb{Maximum Discrete Independent Set (\mdis)} problem and the \colb{Minimum Discrete Dominating Set (\mdds)} problem, respectively. The formal definitions of the \mdis~and \mdds~problems are given below. 

The problems considered in this paper are discrete variants of the \mis~and \mds~problems. We formally define these problems as follows
\begin{siderules}
    \textbf{\textcolor{red!70!blue}{Maximum Discrete Independent Set ({\mdis})}}. Let $\mathcal{R}$ be a set of objects and $\mathcal{P}$ be a set of points in the plane.  Compute a maximum size subset $\mathcal{R}^\prime \subseteq \mathcal{R}$ such that no two objects in $\mathcal{R}^\prime$ cover the same point from  $\mathcal{P}$. 
\end{siderules}

\begin{siderules}
\textbf{\color{red!70!blue}Minimum Discrete Dominating  Set (\mdds)}. Let $\mathcal{R}$ be a set of objects and $\mathcal{P}$ be a set of points in the plane.  Compute a minimum size subset $\mathcal{R}^\prime \subseteq \mathcal{R}$ such that for every object $O \in \mathcal{R} \setminus \mathcal{R}^\prime$,  $O\cap O^\prime\cap {\cal P} \neq \emptyset$ for some $O^\prime \in \mathcal{R}^\prime$. 
\end{siderules}

In this paper, we study the hardness results and polynomial-time approximation schemes\footnote{A polynomial-time approximation scheme (\ptas) is a family of algorithms $\{A_\epsilon\}$, where there is an algorithm for each $\epsilon > 0$, such that $A_\epsilon$ is a $(1 + \epsilon)$-approximation
algorithm (for minimization problems) or a $(1 - \epsilon)$-approximation algorithm (for maximization
problems) \cite{Williamson2011}. The running time is bounded by a polynomial in the size of the instance and $\epsilon$.} (\ptas es) of the \mdis~and \mdds~problems for various geometric objects such as disks, axis-parallel squares, axis-parallel rectangles, and some other shapes. 

We note that the \mdis~and \mdds~problems are at least as hard as the \mis~and \mds~problems, respectively. This can be established by placing a point in each intersection region formed by the given objects in the corresponding instances of the \mis~and \mds~problems.

\subsection{Previous work}
The \mis~problem is known as \np-hard for several classes of objects like unit disks \cite{Clark1990}, unit squares \cite{Fowler1981}, etc. 
Further, \ptas es are also known for unit squares and unit disks \cite{Hunt1998,Matsui2000,Das2015}.   On the other hand, Chan and Har-Peled \cite{Chan2012} gave a \ptas~for the \mis~problem with pseudo-disks based on the local search algorithm.  For axis-parallel rectangles, Adamaszek and Wiese \cite{Adamaszek2014} provided a \qptas. Chuzhoy and Ene \cite{Chuzhoy2016}  also have provided a \qptas~with improved running time. In 2021,  Mitchell \cite{Mitchell21} provided a breakthrough result for the \mis~problem on axis-parallel rectangles and provided a constant factor approximation algorithm. After that, a series of improved constant factor approximation algorithms are reported: factor 6 by Galvez et al. \cite{GalvezKMMPW22}, factor 3 by Galvez et al. \cite{GalvezKMMPW22}, factor $(2+\epsilon)$ by Galvez et al. \cite{GalvezKMMPW22}, factor $\frac{10}{3}$ by Mitchell \cite{Mitchell21}, factor $(3+\epsilon)$ by Mitchell \cite{Mitchell21}.

The \mdis~problem was first studied by Chan and Har-Peled  \cite{Chan2012}. They show that an $\mathsf{LP}$-based algorithm gives an $O(1)$-approximation for pseudo-disks. To our knowledge,  this is the best approximation factor known till now for the \mdis~problem, even for special classes of pseudo-disks like disks, squares, etc. On the other hand,  Chan and Grant \cite{Chan2014} have shown that the \mdis~problem is \apx-hard for various classes of objects like axis-parallel rectangles containing a common point, axis-parallel strips,  ellipses sharing a common point, downward shadows of segments, unit balls in $\mathbb{R}^3$ containing the origin, and other shapes. (see Theorem 1.5 in \cite{Chan2014}).

 The \mds~problem is \np-complete for unit disk graphs \cite{Clark1990}  and a \ptas~is known for the same \cite{Hunt1998}.  Recently, Gibson and Pirwani \cite{Gibson2010} obtained a \ptas~for \mds~problem   for arbitrary radii disks by local search method first used in \cite{Chan2012} and \cite{Mustafa2010}.  However, Erlebach and van Leeuwen \cite{Erlebach2008} have shown that the \mds~problem is \apx-hard for several intersection graphs of objects such as axis-parallel rectangles, ellipses, and other shapes.  Recently, by using local search method,  Bandyapadhyay et al. \cite{Bandyapadhyay2018} gave a $(2 + \epsilon)$ approximation algorithm for the \mds~problem with diagonal-anchored (A set of axis-parallel rectangles is said to be diagonal-anchored, if given a diagonal with slope $-1$ then either the lower-left or the upper-right corner of each rectangle is on the diagonal.) axis-parallel rectangles, for any $\epsilon>0$. They studied $\mathcal{L}$-types of objects which are essentially rectangles when the $\mathcal{L}$-shapes are diagonal-anchored. They gave a local search based \ptas~for a special case where the rectangles are anchored from the same side of the diagonal. 

\subsection{Our contributions}

\begin{itemize}
    \item[\ding{228}] In \cite{Chan2012}, Chan and Har-Peled noted that,  \colb{``Unlike in the original independent set (\mis) problem, it is not clear if the local search yields a good approximation for \mdis~problem, even in the unweighted case"}. In this paper, we answer this partially affirmatively by providing \ptas es for the \mdis~problem with disks and axis-parallel squares. More specifically, we prove the following.

\begin{itemize}
    \item The \mdis~problem admits \ptas es for arbitrary radii disks (Theorem \ref{ptas_dis_disks}) and arbitrary side length squares (Theorem \ref{ptas_dis_squares}).
    
    \item The \mdds~problem admits \ptas es for arbitrary radii disks (Theorem \ref{ptas_dds_disks}) and arbitrary side length squares (Theorem \ref{ptas_dds_squares}).
\end{itemize}

The \ptas es for the \mdis~problem are obtained by extending the local search algorithm given in \cite{Chan2012}. Whereas the \ptas es for the \mdds~problem are obtained by extending the local search method of Gibson and Pirwani \cite{Gibson2010}. \\

\item[\ding{228}] To prove the \apx-hardness results for the \mdds~problem on various objects, we first introduce a special case of the \mds~problem with set systems, the \spds~problem (see Definition \ref{special_3ds_def}) and prove that it is \apx-hard. The proof is inspired by the \apx-hardness result of the \spsc~problem studied by Chan and Grant \cite{Chan2014}. Next, we use the \spds~problem to prove that the \mdds~problem on the following classes of geometric objects are \apx-hard (Theorem \ref{apx-hardness-theorem}). The classes of objects we consider are mentioned in \cite{Chan2014}. 

	\begin{description} 
		
		\item[{\bf A1:}]  Axis-parallel rectangles in $\mathbb{R}^2$, even when all the rectangles have an upper-left corner inside a square with side length $\epsilon$ and lower-right corner inside a square with side length $\epsilon$ for an arbitrary small $\epsilon>0$.
		
		\item[{\bf A2:}] Axis-parallel  ellipses in $\mathbb{R}^2$, even when all the ellipses contain the origin. 
		
		\item[{\bf A3:}] Axis-parallel strips in $\mathbb{R}^2$. 
		\item[{\bf A4:}]   Axis-parallel rectangles in $\mathbb{R}^2$,  even when every pair of the rectangles intersect either zero or four times. 
		\item[{\bf A5:}] Downward shadows of   segments in the plane.

		\item[{\bf A6:}]  Downward shadows of cubic polynomials in the plane. 
		
		\item[{\bf A7:}] Unit ball in $\mathbb{R}^3$, even when the origin is inside every unit ball.
		
		\item[{\bf A8:}] Axis-parallel cubes of similar size in $\mathbb{R}^3$ containing a common point.
		
		\item[{\bf A9:}] Half-spaces in $\mathbb{R}^4$. 
		
		\item[{\bf A10:}] Fat semi-infinite wedges in $\mathbb{R}^2$ with apices near the origin. 
	\end{description}

\begin{comment}

\begin{theorem}
	The \mdds~problem is \apx-hard for the following classes of geometric objects.
	\begin{description} 
		
		\item[{\bf A1}]  Axis-parallel rectangles in $\mathbb{R}^2$, even when all rectangles have upper-left corner inside a square with side length $\epsilon$ and lower-right corner inside a square with side length $\epsilon$ for an arbitrary small $\epsilon>0$.
		
		\item[{\bf A2}] Axis-parallel  ellipses in $\mathbb{R}^2$, even when all the ellipses  contain the origin. 
		
		\item[{\bf A3}] Axis-parallel strips in $\mathbb{R}^2$. 
		\item[{\bf A4}]   Axis-parallel rectangles in $\mathbb{R}^2$,  even when every pair of rectangles intersect either zero or four times. 
		\item[{\bf A5}] Downward shadows of  segments in the plane.

		\item[{\bf A6}]  Downward shadows of   cubic polynomials in the plane. 
		
		\item[{\bf A7}] Unit ball in $\mathbb{R}^3$, even when the origin is inside every unit ball.
		
		\item[{\bf A8}] Axis-parallel cubes of similar size in $\mathbb{R}^3$ containing a common point.
		
		\item[{\bf A9}] Half-spaces in $\mathbb{R}^4$. 
		
		\item[{\bf A10}] Fat semi-infinite wedges in $\mathbb{R}^2$ with apices near the origin. 
	\end{description}
	\label{apx_hardness_mdds}
\end{theorem}
\end{comment}

We note that for classes $\textbf{A1}$-$\textbf{A10}$, the \mdis~problem is known to be \apx-hard \cite{Chan2014}. Further,  in \cite{Chan2014}, authors also have proved that the set cover problem is   \apx-hard for all  classes  of objects $\textbf{A1}$-$\textbf{A10}$ and   hitting  set is  \apx-hard for  four classes of objects $\textbf{A3}, \textbf{A4}, \textbf{A7}$, and $\textbf{A9}$.  Recently,   in \cite{Madireddy2017}, the authors have shown that the hitting set problem is \apx-hard for the remaining classes of objects. We further show that  both \mdis~and \mdds~problems are \apx-hard for $(i)$ fat triangles of similar size\footnote{The diameter of the triangles are in the range  $(2 - \delta, 2]$, for a small $\delta>0$ \cite{HarPeled2009}.},  and $(ii)$ similar circles (see Theorem \ref{more_apx_results}). 

\item[\ding{228}] We also show that both \mdis~and \mdds~problems are \np-hard for unit disks intersecting a horizontal line and axis-parallel unit squares intersecting a straight line of slope $-1$.  Our \np-hardness results are inspired by the results of Fraser and L\'opez-Ortiz \cite{Fraser2017} and Mudgal and Pandit \cite{Mudgal2015}.  We note that in these restricted cases,  \mis~problem can be solved in polynomial time for unit disks \cite{NANDY2017} and unit squares \cite{Mudgal2015}. Further, the \mds~problem can also be solved in polynomial-time for unit squares \cite{Pandit17}. Our \np-hardness results show the gradation of the complexity between continuous and discrete versions of the problems. 

\end{itemize}

\subsection{Organization of the paper}

The remainder of the paper is organized as follows. In Section \ref{ptas_mdis}, we present \ptas es for the \mdis~problem with disks of arbitrary radii and squares of arbitrary side lengths. Section \ref{mdds_ptas} extends these results by providing \ptas es for the \mdds~problem using the same set of objects. Section \ref{apx_hard_results} contains the \apx-hardness results, including the proof of Theorem \ref{apx_hardness_mdds} and related problems. Finally, in Section \ref{np_hard_mdis}, we establish \np-hardness for both the \mdis~and \mdds~problems when restricted to unit disks intersecting a horizontal line and axis-parallel unit squares intersecting a line of slope $-1$.

%The rest of the paper is organized as follows.
% Section \ref{ptas_mdis} presents \ptas es for the \mdis~problem with arbitrary radii disks and arbitrary side length squares. For the same set of objects, we give \ptas es for the \mdds~problem  in Section \ref{mdds_ptas}.    The \apx-hardness results (proof of  Theorem \ref{apx_hardness_mdds} and other related problems) are presented in Section \ref{apx_hard_results}. Finally, we present the \np-hardness results for both \mdis~and \mdds~problems with unit disks intersecting a horizontal line and axis-parallel unit squares intersecting a straight line with slope −1 in Section \ref{np_hard_mdis}.

%The proofs of \apx-hardness results in   \cref{apx_hardness_mdds} and  few more \apx-hardness results of both problems are presented in \cref{apx_hard_results}. 

\section{\texorpdfstring{\boldmath{\ptas}}~: Maximum Discrete Independent Set Problem} \label{ptas_mdis}
This section presents \ptas es for the \mdis~problem with arbitrary radii disks and arbitrary side length squares. These \ptas es are obtained by extending the local search technique of Chan and Har-Peled \cite{Chan2012} for the \mis~problem with pseudo-disks.   

\subsection{The algorithm}

Let $(\mathcal{P}, \mathcal{R})$ be the input to the \mdis~problem where $\mathcal{P}$ is a set of points, and  $\mathcal{R}$ is a set of objects in the plane. Further, let $m=|\mathcal{R}|$ and $n= |\mathcal{P}|$. Without loss of generality, we assume that no object completely covers another object in $\mathcal{R}$. A set $\mathcal{L} \subseteq \mathcal{R}$ is a \colb{feasible solution} to the \mdis~problem if no two objects in $\mathcal{L}$ cover the same point from $\mathcal{P}$. 
For a given integer $t > 1$, a feasible solution $\mathcal{L}$ is  \colb{$t$-locally optimal} if we cannot obtain another feasible solution  $\mathcal{L}^\prime \subseteq \mathcal{R}$ of larger size,  by replacing at most $t$ objects from $\mathcal{L}$  with at most $t+1$ objects from $\mathcal{R}$.

Algorithm \ref{localsearchalgo_dis} describes the procedure to compute a $t$-locally optimal solution for the \mdds~problem.
%We now describe the procedure to obtain a $t$-locally optimal solution to the \mdis~problem in Algorithm \ref{localsearchalgo_dis}. 
Note that, in every local exchange (step 5), the size of $\mathcal{L}$ is increased by at least one. Hence, the local exchange can be possible at most $m$ times. However, every such step needs to go over all possible sets $\mathcal{R}^\prime$ and $\mathcal{L}^\prime$.
Since $|\mathcal{R}^\prime| \leq t+1$, there are $O(m^{t+1})$ possibilities for its value.
Similarly, there are $O(m^t)$ different possible values for $|\mathcal{L}^\prime|$. Checking whether
$(\mathcal{L} \setminus \mathcal{L}^\prime) \cup \mathcal{R}^\prime$  is feasible requires $O(nm)$ time.
%Since $|\mathcal{R}^\prime| \leq t+1$, there will be at most $O(m^{t+1})$  possibilities for $\mathcal{R}^\prime$  and for every such  $\mathcal{R}^\prime$ at most $O(m^t)$ number of different $\mathcal{L}^\prime$ are possible. Further, to check whether $(\mathcal{L} \setminus \mathcal{L}^\prime) \cup \mathcal{R}^\prime$  is a feasible solution or not, one needs $O(nm)$-time. 
Hence, Algorithm \ref{localsearchalgo_dis} returns a $t$-locally optimal solution $\mathcal{L} \subseteq \mathcal{R}$ in  $O(nm^{2t+3})$-time.

%\begin{algorithm}[ht!]
%	\caption{$t$-level local search for \mdis~problem}
	%\label{localsearchalgo_dis}
%	\begin{algorithmic}[1]
	%	\State Let $\mathcal{L} \gets \emptyset$. 
		
%		\For{$\mathcal{R}^\prime \subseteq \mathcal{R} \setminus \mathcal{L}$ of size at most $t+1$}
		
%		\For{ $\mathcal{L}^\prime \subseteq \mathcal{L}$ of size at most $t$}
%		\If{$(\mathcal{L} \setminus \mathcal{L}^\prime) \cup \mathcal{R}^\prime$ is a  feasible solution and $|(\mathcal{L} \setminus \mathcal{L}^\prime) \cup \mathcal{R}^\prime| \geq |\mathcal{L}| + 1$ }
	%	\State $\mathcal{L} \gets (\mathcal{L} \setminus \mathcal{L}^\prime) \cup \mathcal{R}^\prime$ \Comment{local exchange step}
%		\EndIf
%		\EndFor
%	\EndFor
%	\end{algorithmic}
%\end{algorithm}

\begin{algorithm}[H]
\caption{$t$-level local search for \mdis~problem}
\label{localsearchalgo_dis}
\SetAlgoLined
\DontPrintSemicolon
$\mathcal{L} \gets \emptyset$\;
\For{$\mathcal{R}^\prime \subseteq \mathcal{R} \setminus \mathcal{L}$ of size at most $t+1$}{
    \For{$\mathcal{L}^\prime \subseteq \mathcal{L}$ of size at most $t$}{
        \If{$(\mathcal{L} \setminus \mathcal{L}^\prime) \cup \mathcal{R}^\prime$ is a feasible solution and $|(\mathcal{L} \setminus \mathcal{L}^\prime) \cup \mathcal{R}^\prime| \geq |\mathcal{L}| + 1$}{
            $\mathcal{L} \gets (\mathcal{L} \setminus \mathcal{L}^\prime) \cup \mathcal{R}^\prime$   \tcp*{local exchange step} 
        }
    }
}
\end{algorithm}

% \begin{algorithm}
% \caption{$t$-level local search for \mdis~problem}
% \label{localsearchalgo_dis}
% \begin{algorithmic}[1]
% \STATE{Let $\mathcal{L} \gets \emptyset$.}
% \FOR{$\mathcal{R}^\prime \subseteq \mathcal{R} \setminus \mathcal{L}$ of size at most $t+1$}
%     \FOR{$\mathcal{L}^\prime \subseteq \mathcal{L}$ of size at most $t$} 
%         \IF{$(\mathcal{L} \setminus \mathcal{L}^\prime) \cup \mathcal{R}^\prime$ is a  feasible solution and $|(\mathcal{L} \setminus \mathcal{L}^\prime) \cup \mathcal{R}^\prime| \geq |\mathcal{L}| + 1$}
            
%                 \STATE $\mathcal{L} \gets (\mathcal{L} \setminus \mathcal{L}^\prime) \cup \mathcal{R}^\prime$ ~~~~ {//local exchange step}
            
%             \ENDIF
%     \ENDFOR
% \ENDFOR
% \end{algorithmic}
% \end{algorithm}

%\begin{algorithm}[ht!]
%  \SetAlgoLined
%\STATE   Let $\mathcal{L} \gets \emptyset$.

%   \FOR{$\mathcal{R}^\prime \subseteq \mathcal{R} \setminus \mathcal{L}$ of size at most $t+1$}
   
%        \FOR{$\mathcal{L}^\prime \subseteq \mathcal{L}$ of size at most $t$}
        
%            \IF{$(\mathcal{L} \setminus \mathcal{L}^\prime) \cup \mathcal{R}^\prime$ is a  feasible solution and $|(\mathcal{L} \setminus \mathcal{L}^\prime) \cup \mathcal{R}^\prime| \geq |\mathcal{L}| + 1$}
            
%                \STATE $\mathcal{L} \gets (\mathcal{L} \setminus \mathcal{L}^\prime) \cup \mathcal{R}^\prime$ ~~~~~~~~~ \tcc{local exchange step}
            
%            \ENDIF
        
%        \ENDFOR

%   \ENDFOR

%  \caption{$t$-level local search for \mdis~problem}
%  \label{localsearchalgo_dis}
%\end{algorithm}

In the following, we first show that Algorithm \ref{localsearchalgo_dis} returns a $t$-locally optimal solution that has the size at least $(1-O(\frac{1}{\sqrt{t}}))$ times the size of the optimal solution to the \mdis~problem  when the objects are arbitrary radii disks. Later, we show that the same is also true for the arbitrary side length axis-parallel squares. We run the above algorithm with $t=O(1/\epsilon^2)$ to provide the desired $(1+\epsilon)$-approximation.

\subsection{Preliminaries} \label{premilinaries}

Assume that $\mathcal{R}$ is a set of arbitrary radii disks. Without loss of generality, we assume that no three disk centers and points in $\mathcal{P}$ are collinear, and no more than three disks are tangent to a circle \cite{Gibson2010,Wan2011}.  For a  disk $D$, let  $\mathsf{cen}(D)$ and $\mathsf{radius}(D)$ denote the center and radius of $D$ respectively. Let $||x-y||$ denote the Euclidean distance between the points $x$ and $y$ in the plane. 

\begin{definition}
For a disk $D$ and a point $p$ in the plane, we define  $\mathsf{\Phi}(D, p) = ||\mathsf{cen}(D)-p||-\mathsf{radius}(D)$.

\end{definition}

%$\mathsf{\Phi}(D, p)$ as the length of the shortest line segment between a point on the boundary of $D$ and $p$, i.e.,

For the given instance $(\mathcal{P}, \mathcal{R})$ of the \mdis~problem, let
 $\mathcal{L} \subseteq \mathcal{R}$ be the $t$-locally optimal  solution return by  Algorithm \ref{localsearchalgo_dis}  and  let $\mathcal{O} \subseteq \mathcal{R}$ be an optimal solution.

One can assume that $\mathcal{L} \cap \mathcal{O} = \emptyset$. To see this, we follow the argument of Mustafa and Ray \cite{Mustafa2010}.
%Without loss of generality, assume that $\mathcal{L} \cap \mathcal{O} = \emptyset$ (see \cite{Mustafa2010} for  details). 
Suppose that this statement is not true. Let $\mathcal{T}  = \mathcal{L} \cap \mathcal{O}$, $\mathcal{L}^* = \mathcal{L} \setminus \mathcal{T}$, $\mathcal{O}^* = \mathcal{O} \setminus \mathcal{T}$. Further, let $\mathcal{P}^* \subseteq \mathcal{P}$ be the set of points that are not covered by any disk in $\mathcal{T}$ and $\mathcal{R}^* \subseteq \mathcal{R}$ be the set of disks that are independent of the disks in $\mathcal{T}$. No disk in $\mathcal{R}^*$ covers points in $\mathcal{P}^*$.  Note that $\mathcal{L}^*$ and $\mathcal{O}^*$ are disjoint. Further, $\mathcal{O}^*$ is an independent set of maximum size for the discrete independent set problem for $\mathcal{P}\setminus  \mathcal{P}^*$ and $\mathcal{R}^*$. Therefore, if $|\mathcal{L}^*| \geq (1 - \epsilon) |\mathcal{O}^*|$ (for a small $\epsilon >0$), then $|\mathcal{L}| \geq (1 - \epsilon) |\mathcal{O}|$. Furthermore, any beneficial $t$-local exchange for $\mathcal{L}^*$ and $\mathcal{P}^*$ is a beneficial $t$-local exchange for $\mathcal{L}$ and $\mathcal{P}$. Thus, we can apply our analysis to $\mathcal{L}^*$ and $\mathcal{P}^*$. Hence, in the rest of the section, we assume that $\mathcal{L} \cap \mathcal{O} = \emptyset$.

 For a disk $D \in \mathcal{L} \cup \mathcal{O}$, let $\mathsf{cell}(D)$ be the set of points $p$ in the plane such that $\mathsf{\Phi}(D, p) \leq \mathsf{\Phi}(D^\prime, p) \text{ for all } D^\prime \in \mathcal{L} \cup \mathcal{O}$,  i.e., $\mathsf{cell}(D) = \{p \mid \mathsf{\Phi}(D, p) \leq \mathsf{\Phi}(D^\prime, p),  ~\forall D^\prime \in \mathcal{L} \cup \mathcal{O}\}$.  The collection of all cells of disks  in $\mathcal{L} \cup \mathcal{O}$ defines the \colb{Additive Weighted Voronoi Diagram $(\mathsf{AWVD})$}, i.e., $ \mathsf{AWVD} = \bigcup_{D \in \mathcal{L} \cup \mathcal{O}} \mathsf{cell}(D)$. We use $\seg(p, q)$ to denote the line segment with endpoints $p$ and $q$. 
 
 We now mention two properties of cells in the $\mathsf{AWVD}$.

\begin{lemma}[\cite{Gibson2010}] 
	The following two properties are true for each disk $D$ in any set of disks such that no disk is contained inside another disk.
	\begin{description}
		\item[{\bf I:}] $\mathsf{cell}(D)$ is non-empty. In particular,  the point $\mathsf{cen}(D)$ is contained only in $\mathsf{cell}(D)$.
		\item[{\bf II:}] $\mathsf{cell}(D)$ is star-shaped, i.e., for any point $p \in \mathsf{cell}(D)$, every point on  the segment $\seg(p,\mathsf{cen}(D))$ is in $\mathsf{cell}(D)$. 
	\end{description}
	
In particular, these properties hold for the set $\mathcal{L} \cup \mathcal{O}$.
	
	\label{startproperty}
\end{lemma}

\begin{proof}
 ({\bf I}): Assume that $\mathsf{cen}(D)$ is contained in  $\mathsf{cell}(D^\prime)$ for some disk $D^\prime$ such that $D^\prime \neq D$. Then, $\mathsf{\Phi}(D^\prime, \mathsf{cen}(D)) \leq \mathsf{\Phi}(D, \mathsf{cen}(D)) = ||\mathsf{cen}(D) - \mathsf{cen}(D)|| - \mathsf{radius}(D) = - \mathsf{radius}(D)$. Thus, $||\mathsf{cen}(D^\prime) - \mathsf{cen}(D)|| - \mathsf{radius}(D^\prime) \leq - \mathsf{radius}(D)$ which implies 
$||\mathsf{cen}(D^\prime) - \mathsf{cen}(D)|| +  \mathsf{radius}(D) \leq  \mathsf{radius}(D^\prime)$. 
%Since $D \neq D^\prime$ and $||\mathsf{cen}(D^\prime) - \mathsf{cen}(D)|| +  \mathsf{radius}(D) \leq  \mathsf{radius}(D^\prime)$, 
Since $D \neq D^\prime$, the disk $D$ is completely contained inside $D^\prime$, and this contradicts the assumption that no disk in $\mathcal{R}$ is completely contained inside any other disk in $\mathcal{R}$. Hence, $\mathsf{cen}(D)$ is only in $\mathsf{cell}(D)$. 

\smallskip
\noindent ({\bf II}): Let $x$ be a point on $\seg(\mathsf{cen}(D),p)$ such that $x \in \mathsf{cell}(D^\prime)$ for some $D^\prime \in \mathcal{L} \cup \mathcal{O}$ and $D \neq D^\prime$. Then, $\mathsf{\Phi}(D^\prime, x) \leq \mathsf{\Phi}(D, x)$.

We have,  $||\mathsf{cen}(D^\prime) - p||$  $ \leq ||\mathsf{cen}(D^\prime) - x || + ||x - p||$. By subtracting 
$\mathsf{radius}(D^\prime)$ from both sides of this inequality we have,
$||\mathsf{cen}(D^\prime) - p|| - \mathsf{radius}(D^\prime)  \leq ||\mathsf{cen}(D^\prime) - x|| + ||x - p|| - \mathsf{radius}(D^\prime)$. 
This implies that 
$\mathsf{\Phi}(D^\prime, p) \leq \mathsf{\Phi}(D^\prime, x)  +  ||x - p|| \leq \mathsf{\Phi}(D, x)  +  ||x -p || \leq  \mathsf{\Phi}(D, p)$.

 Thus, $p$ also belongs to $\mathsf{cell}(D^\prime)$, which is not possible. \hfill   
\end{proof}

\begin{lemma}[\cite{Wan2011}]\label{properties_WVD}
	Let $D_1$ and $D_2$ be two disks in $\mathcal{L} \cup \mathcal{O}$. Let $x$ be a point in the plane such that $\mathsf{\Phi}(D_1, x) \leq \mathsf{\Phi}(D_2, x)$. If $D_2$ covers $x$, then $D_1$ also covers $x$.
\end{lemma}

\begin{proof}
If $D_2$ covers the point $x$, then $\mathsf{\Phi}(D_2, x) \leq 0$. This implies that $\mathsf{\Phi}(D_1, x) \leq 0$. Thus, $||\mathsf{cen}(D_1) - x||-\mathsf{radius}(D_1) \leq 0$. Hence, the disk $D_1$ also covers the point $x$. \hfill  
\end{proof}

 Let $G = (V, E)$ be a given graph. For a vertex $v \in V$, let $N(v)$ be the set of vertices adjacent to $v$ in $G$. For a subset $V^\prime \subseteq V$ of vertices, let  $N(V^\prime)$ be the set of all adjacent vertices of the vertices in $V^\prime$, i.e., $N(V^\prime) = \bigcup_{v \in V^\prime} N(v)$. Further, let ${N^+(V^\prime)} = N(V^\prime) \cup V^\prime$. We need the following result that is implied by the planar separator theorem \cite{Frederickson1987}.

 %We now note  a planar separator theorem from \cite{Frederickson1987} which is required in proving the  performance of the local search algorithm. 

\begin{lemma}[\cite{Frederickson1987}] \label{planar_sepa_lemma}
	For a given  planar graph $G = (V, E)$ and a parameter $r \geq 1$, there exists  a subset $X \subseteq V$ of size at most $c_1 |V| / \sqrt{r}$, and a division of $V \setminus X$ into $|V|/r$ sets $V_1, V_2, \ldots, V_{|V|/r}$ such that $(i) ~ |V_i| \leq c_2 r$, $(ii) ~ N(V_i) \cap V_j = \emptyset$ for $i \neq j$, and $(iii) ~ |N(V_i) \cap X| \leq c_3 \sqrt{r}$ for some constants $c_1, c_2$, and $c_3$. 
\end{lemma}

\subsection{Analysis of the algorithm}

In this section, we show that the $t$-level local search in Algorithm \ref{localsearchalgo_dis} is a \ptas~for the \mdis~problem with arbitrary radii disks and arbitrary side length axis-parallel squares. 

    \begin{theorem}
The \mdis~problem with arbitrary radii disks admits a \ptas, i.e., For any integer $t > 1$,  the $t$-level local search in Algorithm \ref{localsearchalgo_dis} produces a solution $\mathcal{L}$ of size $\geq (1-\epsilon) |\mathcal{O}|$,  $\epsilon = O(\frac{1}{\sqrt{t}})$ in $O(nm^{O(1/\epsilon^2)})$-time, where $\cal O$ is an optimum solution of the \mdis~problem, $m$ is the number of disks, and $n$ is the number of points.  

  %Algorithm \ref{localsearchalgo_dis} returns we prove that the size of the optimal solution $\mathcal{O}$ is at most $(1+O(\frac{1}{\sqrt{t}})$ times the size of the  local solution  $\mathcal{L}$ returned by , for any $\epsilon > 0$.   
   \label{ptas_dis_disks}
\end{theorem}

\begin{proof}
%{Proof of Theorem \ref{ptas_dis_disks}} 

The main idea of the proof is to construct a planar bipartite graph $G = (V, E)$, where each node in $V$ corresponds to each disk in $\mathcal{L} \cup \mathcal{O}$ and the edges connect the vertices corresponding to the disks in $\mathcal{L}$ and $\mathcal{O}$ that satisfy some properties.  Next, we apply Lemma \ref{planar_sepa_lemma} on $G$. We show that for any $V_i$ (part of the result of Lemma \ref{planar_sepa_lemma}), $|\mathcal{O} \cap V_i| \leq |\mathcal{L}\cap V_i| + |N(V_i) \cap X|$. This gives us the required relationship between  $|\mathcal{L}|$ and $|\mathcal{O}|$.

We define a graph $G = (V, E)$, which can be viewed as a subgraph of the dual of  $\mathsf{AWVD}$ of disks in $\mathcal{L} \cup \mathcal{O}$. Recall that  $\mathcal{L} \cap \mathcal{O} = \emptyset$. 
\begin{enumerate}
	\item For every disk $D \in \mathcal{L} \cup \mathcal{O}$, we place a vertex in $G$  at $\mathsf{cen}(D)$.
	
	\item For every $L \in \mathcal{L}$ and $O \in \mathcal{O}$,  place an edge between $\mathsf{cen}(L)$ and $\mathsf{cen}(O)$ if $\mathsf{cell}(L)$ and $\mathsf{cell}(O)$ share a common boundary. 
\end{enumerate}

By using the star-shaped property of cells, one can draw $G$ such that no two edges intersect \cite{Aurenhammer1991}. Thus, the graph $G = (V, E)$ is planar and bipartite.

%As in \cite{Chan2012}, w
We now apply Lemma \ref{planar_sepa_lemma} on the graph $G$ with  $r = t /(c_2 + c_3)$, where $c_2,c_3$ are the constants as in Lemma \ref{planar_sepa_lemma}. Then, $|{N^+(V_i)}| \leq |V_i| +  |N(V_i)| \leq c_2 r + c_3 \sqrt{r} \leq  (c_2 + c_3)  r \leq  t$. For each $i$,   let $\mathcal{O}_i = V_i \cap \mathcal{O}$, $\mathcal{L}_i = V_i \cap \mathcal{L}$, and $X_i = N(V_i) \cap X$, where $X\subseteq V$ is the same as defined in Lemma \ref{planar_sepa_lemma}.  

We now prove that  $\mathcal{Y}_i = (\mathcal{L} \setminus {N^+(V_i)}) \cup  \mathcal{O}_i$  is a feasible solution for the \mdis~problem.  We note that any subset of a feasible solution is also a feasible solution of the \mdis~problem. Hence, $\mathcal{L} \setminus {N^+(V_i)}$ and  $\mathcal{O}_i$ are also  feasible solutions of the \mdis~problem. For the sake of contradiction, assume that  $\mathcal{Y}_i$  is not a feasible solution. Hence, there exist two disks $O \in \mathcal{O}_i$ and $L \in (\mathcal{L} \setminus {N^+(V_i)})$ such that both $O$ and $L$ cover the same point $p \in \mathcal{P}$. Note that, $O$ and $L$ are the unique disks in $\mathcal{O}_i$ and $\mathcal{L} \setminus {N^+(V_i)}$, respectively, that cover the point $p$. We consider that $p \in \mathsf{cell}(O)$ (the case when $p \in \mathsf{cell}(L)$, a similar argument can be given).  Since $p \in \mathsf{cell}(O)$, $\mathsf{\Phi}(O, p) \leq  \mathsf{\Phi}(L, p)$. There are two possible cases. (The argument follows the proof of Lemma 3 in \cite{Gibson2010}.)

%for the two cases are on the same lines of the proof of Lemma 3 in \cite{Gibson2010}.)

\begin{description}
\item[\it Case 1:] Suppose $\mathsf{\Phi}(O, p) = \mathsf{\Phi}(L, p)$. Then $p \in \mathsf{cell}(L)$. Hence, $\mathsf{cell}(O)$ and $\mathsf{cell}(L)$ share a common boundary in $\mathsf{AWVD}$ and further, $O \in \mathcal{O}$ and $L \in \mathcal{L}$. Thus,  there exists an edge between $\mathsf{cen}(O)$ and $\mathsf{cen}(L)$ in graph $G$. Hence  $L \in N(O)$ which implies  $L \notin \mathcal{L} \setminus {N^+(V_i)}$.

\item[\it Case 2:] Suppose $\mathsf{\Phi}(O, p) < \mathsf{\Phi}(L, p)$. Take a walk from $p$ to $\mathsf{cen}(L)$ along the line segment 
 $\seg(p, \mathsf{cen}(L))$. Note that the segment may cross several cells. Let $q$ be the point on this segment entering $\mathsf{cell}(L)$ (see Fig.  \ref{caas_b}). Therefore, $\mathsf{\Phi}(D, q) = \mathsf{\Phi}(L, q)$ for  some $D \in \mathcal{L} \cup \mathcal{O}$. 
Define $B_i = N^+(V_i) \setminus V_i$ to be the boundary of $i$-th patch. Since $O \in V_i$ and $L$ are outside $V_i \cup B_i$, they are not connected in $G$. Thus, if $D = O$, then $O$ and $L$ would be connected in $G$, which is impossible.  Therefore, we assume that $D \neq O$. We now prove that $D$ covers $p$. Since no three disk centers and points in $\mathcal{P}$ are collinear, we have	$||\mathsf{cen}(D) -  p||  < ||p - q || + ||\mathsf{cen}(D) -  q||$, this implies that $\mathsf{\Phi}(D, p) < ||p - q|| + \mathsf{\Phi}(D, q)$  $ = ||p - q|| + \mathsf{\Phi}(L, q) = \mathsf{\Phi}(L, p)$.  	Since  $\mathsf{\Phi}(D, p) < \mathsf{\Phi}(L, p) $ and $L$ covers $p$, from Lemma \ref{properties_WVD} we have that $D$ also covers $p$.  Suppose $D \in \mathcal{O}$. Then, $p$ is covered by the two disks $O$ and $D$ in   $\mathcal{O}$, which is impossible.   Suppose $D \in \mathcal{L}$. In this case, $p$ is also covered by the two disks  $D$ and $L$ in  $\mathcal{L}$, which is impossible. 
		
	\end{description}
	
	\begin{figure}[ht!]
 	    \begin{center}
 	    \includegraphics[scale=0.7]{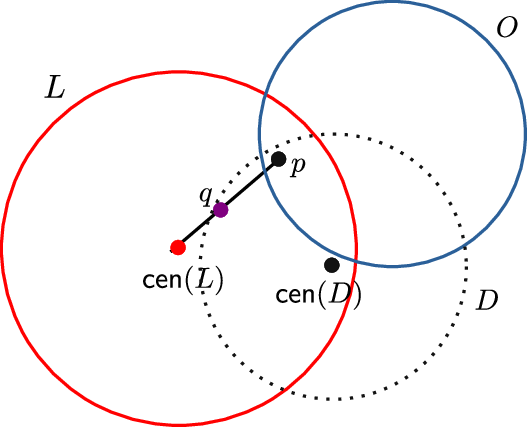}
 	    \caption{An  illustration of an existence of  a point $q$ on the segment $\seg(p, \mathsf{cen}(L))$. }
 	    \label{caas_b}    
 	    \end{center}
 	    
 	\end{figure}

Therefore, $\mathcal{Y}_i$  is a feasible solution for the \mdis~problem.

We now proceed as in \cite{Chan2012}. If  $|\mathcal{O}_i| > |\mathcal{L}_i| + | X_i|$, then by replacing disks of $\mathcal{L} \cap {N^+(V_i)} $ in $\mathcal{L}$ with disks in $\mathcal{O}_i$, we get  a better solution. This contradicts the $t$-local optimality of $\mathcal{L}$. Hence, $|\mathcal{O}_i| \leq  |\mathcal{L}_i| + | X_i|$.  Thus, 
\begin{center}
	\begin{tabular} {l l }
		$|\mathcal{O} |$ & $\leq \Sigma_i |\mathcal{O}_i| + |X|$  $\leq  \Sigma_i |\mathcal{L}_i| + \Sigma_i |X_i| + |X|$ \\
		& $\leq |\mathcal{L}| + c_3 \sqrt{r} \frac{|V|}{r} + c_1 \frac{|V|}{\sqrt{r}}$ $ \leq |\mathcal{L}| + (c_1 + c_3) \frac{|V|}{\sqrt{r}}$\\
 &$= |\mathcal{L}| + (c_1 + c_3) \frac{|\mathcal{O}| + |\mathcal{L}|}{\sqrt{r}} $.

	\end{tabular}
\end{center}

Recall that, we have $r = t / (c_2 + c_3)$. Substituting this value of $r$ in the above inequality, we get \\
\centerline{$|\mathcal{O} |  \leq |\mathcal{L}| + c \frac{|\mathcal{O}| + |\mathcal{L}|}{\sqrt{t}} $, where $c = (c_1 +c_3) \sqrt{c_2+c_3}$~. } 

%Assuming $t \ge 4 c^2$ and set $c^\prime  = 4c$ (similar to the proof of Lemma 2.3 in \cite{Mustafa2010}), we have,

Assuming $t \ge 4 c^2$ and set $c^\prime  = 4c$  we have,
\begin{center}
	\begin{tabular} {l l }
		$|\mathcal{O}|$ & $ \leq  |\mathcal{L}| \frac{1+c/\sqrt{t}}{1-c/\sqrt{t}} =  |\mathcal{L}|(1+c/\sqrt{t}) (1 + (  c/\sqrt{t}) + (  c/\sqrt{t})^2 + \ldots  )$ \\
		& $\leq |\mathcal{L}| (1+c/\sqrt{t}) (1 + 2 c/\sqrt{t})$ ~~~~~~ since $c/\sqrt{t} \leq 1/2     $\\
		& $=  |\mathcal{L}| (1+3c/\sqrt{t} +  2 c^2/{t})$ \\
       & $\leq  |\mathcal{L}| (1+4c/\sqrt{t})$ ~~~~~~~~~~~~since $2 c^2/{t} \leq c/\sqrt{t}$   \\
      & $ = |\mathcal{L}| (1+c^\prime/\sqrt{t})$ \\
	\end{tabular}
\end{center}
%$|\mathcal{O} |$  $ \leq |\mathcal{L}| \frac{1+c/\sqrt{t}}{1-c/\sqrt{t}} =  |\mathcal{L}|(1+c^\prime/\sqrt{t})$.

This implies that $|\mathcal{O}| \leq \bigl(1 + O(\frac{1}{\sqrt{t}})\bigl)|\mathcal{L}|$. Hence, we conclude that, by choosing  $t = O(1/\epsilon^2)$, the local search given in Algorithm \ref{localsearchalgo_dis} gives a $(1-\epsilon)$-approximation algorithm for the \mdis~problem with disks.     \hfill   
\end{proof}

\begin{theorem}
    
The \mdis~problem with arbitrary side length axis-parallel squares admits a \ptas, i.e., For any integer $t > 1$,  the $t$-level local search in Algorithm \ref{localsearchalgo_dis} produces a solution $\mathcal{L}$ of size $\geq (1-\epsilon) |\mathcal{O}|$,  $\epsilon = O(\frac{1}{\sqrt{t}})$ in $O(nm^{O(1/\epsilon^2)})$-time, where $\cal O$ is an optimum solution of the \mdis~problem, $m$ is the number of squares of arbitrary side length, and $n$ is the number of points.  

\label{ptas_dis_squares}

\end{theorem}

%The analysis is similar to the analysis of arbitrary radii disks, except that for any two points $p$ and $q$ in the plane, $||p - q||$ is defined under infinity norm $L_\infty$ instead of $L_2$ norm as in  \cite{Aschner2013}.

\begin{proof}%{Proof sketch of Theorem \ref{ptas_dis_squares}}
Let $\mathcal{R}$ be the set of axis-parallel squares with arbitrary side lengths. Apply $t$-level local search given in Algorithm \ref{localsearchalgo_dis} with $t = O( 1/\epsilon^2)$. The analysis is similar to the analysis of arbitrary radii disks,  except that for any two points $p$ and $q$ in the plane, $||p - q||$ is defined with respect to the infinity norm ${L}_\infty$ instead of $L_2$-norm as in \cite{Aschner2013}. For any square $S$, the center and the side length  of $S$ are denoted by $c_S$ and $l_S$, respectively. For a given point $x$ in the plane and a square $S$, we define  $\mathsf{\Phi}(S, x) = \norm{c_S, x}_\infty - \frac{l_S}{2}$. In particular, if $x$ is on the boundary of $S$ then  $\mathsf{\Phi}(S, x) = 0$, if $x$ is inside $S$ then $\mathsf{\Phi}(S, x) $ is negative, and if $x$ is outside $S$ then $\mathsf{\Phi}(S, x)$ is positive. For any given square $S$, the definition of $\mathsf{cell}(S)$ is the same as the case for disks above.   For the given instance $(\mathcal{P}, \mathcal{R})$ of the \mdis~problem (where the objects in $\mathcal{R}$ are axis-parallel squares), let  $\mathcal{L} \subseteq \mathcal{R}$ be the $t$-locally optimal  solution return by  Algorithm \ref{localsearchalgo_dis}  and  let $\mathcal{O} \subseteq \mathcal{R}$ be an optimal solution. The Additive Weighted Voronoi Diagram $(\mathsf{AWVD})$ is defined as the union of all the cells for squares in $\mathcal{L} \cup \mathcal{O}$. We note that, both Lemma \ref{startproperty} and \ref{properties_WVD} remain true for the kind of $\mathsf{AWVD}$ that we defined for squares. Further, in the same lines of the proof of Theorem \ref{ptas_dis_disks}, we can prove that $|\mathcal{O}| \leq (1 + O(\frac{1}{\sqrt{t}}))|\mathcal{L}|$, where $t$ is the local search parameter. By choosing  $t = O(1/\epsilon^2)$, Algorithm \ref{ptas_dis_disks} gives a $(1-\epsilon)$-approximation algorithm for the \mdis~problem with axis-parallel squares.   \hfill   \end{proof}

\section{\texorpdfstring{\boldmath{\ptas}:}~ Minimum Discrete Dominating Set Problem } \label{mdds_ptas}
In this section, we first give a \ptas~for the \mdds~problem with arbitrary radii disks by using a local search algorithm similar to \cite{Gibson2010}. Further, we show that the same local search algorithm will give a \ptas~for the \mdds~problem with arbitrary side length axis-parallel squares. 

\subsection{The algorithm}
 $\mathcal{P}$ be a  set of $n$ points and  $\mathcal{R}$ be a set of $m$ disks in the plane. Let $\mathcal{R}_\mathsf{cen} $ be the set of all the centers of the disks in $\mathcal{R}$. We assume that no three points from $\mathcal{R}_\mathsf{cen} \cup \mathcal{P}$ are collinear and no more than three disks are tangent to a circle \cite{Gibson2010,Wan2011}. Further,  without loss of generality, assume that no point in $\mathcal{P}$ lies on the boundary of a disk in $\mathcal{R}$. If not, one can slightly perturb the plane such that the assumption becomes true (see \cite{Gibson2010}).  

 Let $D$ and $D^\prime$ be the two disks in $\mathcal{R}$ such that both $D$ and $D^\prime$ cover a point $p \in \mathcal{P}$, then we say that $D$ is a \textit{dominator} of $D^\prime$ and vice versa.    A set $\mathcal{R}^\prime \subseteq \mathcal{R}$ of disks is said to be a  \textit{feasible solution} to the \mdds~problem, if for every disk $O \in (\mathcal{R} \setminus \mathcal{R}^\prime)$, there exists at least one dominator in  $ \mathcal{R}^\prime$.    For a given integer $t > 1$, we say that a feasible solution   $\mathcal{L} \subseteq \mathcal{R}$ is   $t$-\textit{locally optimal} if one cannot obtain a smaller size feasible solution $\mathcal{L}^\prime \subseteq \mathcal{R}$  by replacing at most $t$ disks from $\mathcal{L}$  with at most $t-1$ disks from $\mathcal{R}$. One can obtain a $t$-locally optimal solution to the \mdds~problem by using a local search method similar to  Algorithm \ref{localsearchalgo_dis}. Set $\mathcal{L} \gets \mathcal{R}$. For $\mathcal{L}^\prime \subseteq \mathcal{L}$ of size at most $t$ and for every $\mathcal{R}^\prime \subseteq \mathcal{R} \setminus \mathcal{L}$ of size at most $t-1$, verify whether $(\mathcal{L} \setminus \mathcal{L}^\prime) \cup \mathcal{R}^\prime$ is a  feasible solution and $|(\mathcal{L} \setminus \mathcal{L}^\prime) \cup \mathcal{R}^\prime| \leq |\mathcal{L}| -1$. If yes, replace  $\mathcal{L}$ with $(\mathcal{L} \setminus \mathcal{L}^\prime) \cup \mathcal{R}^\prime$ (local exchange). Repeat this procedure until no further local exchange is possible. Further, the procedure returns a $t$-locally optimal solution in  $O(nm^{2t+3})$-time. 

\subsection{Preliminaries}

Let $\mathcal{L} \subseteq \mathcal{R}$ be a $t$-locally optimal solution returned by the local search algorithm, and let $\mathcal{O} \subseteq \mathcal{R}$ be an optimal solution for the \mdds~problem. We can modify the solution $\mathcal{L}$ such that no disk $L$ in $\mathcal{L}$ that covers a set of points which is a proper subset of the set of points covered by any disk $D$ in $\mathcal{R}$. If this is not the case, then $L$ is replaced by the disk $D$ in $\mathcal{L}$, i.e., the modified $\mathcal{L}$ becomes $(\mathcal{L}\setminus \{L\} ) \cup \{D\}$.
The modified $\mathcal{L}$ is still a solution, i.e., a discrete dominating set for the input $({\cal R}, {\cal P})$ as $D$ dominates the same set (and possibly more) of disks as $L$ dominates. The size of the modified solution is at most the size of the old solution.  
%Further, note that no subset of $\mathcal{L}$ with size at most $t$ can be replaced with a subset of $\mathcal{R}$ of size at most $t-1$. Hence, no local improvement is possible. The above replacement does not affect the size of the solution $\mathcal{L}$. This ensures that, in the modified $\mathcal{L}$,  for no disk  $\mathcal{L}$,  the point set covered by the disk is a proper subset of a point set covered by  another disk in $\mathcal{R}$. 
By applying a similar argument, we assume that no disk $O$ in $\mathcal{O}$ that covers a set of points, which is a proper subset of the set of points covered by any disk $D$ in $\mathcal{R}$.

One can assume that $\mathcal{L} \cap \mathcal{O} = \emptyset$ by applying the argument given in Section \ref{premilinaries}. 
We now have the following {\colb{Locality Condition}} \cite{Gibson2010}, that is used to prove the existence of a \ptas.

\begin{lemma}[\textbf{{Locality Condition} \cite{Gibson2010}}]:
	There exists a planar bipartite graph $G = (\mathcal{L} \cup \mathcal{O}, E)$  such that for every object $X \in \mathcal{R}$, there exists an edge between $L \in \mathcal{L}$ and $O \in \mathcal{O}$ where both $L$ and $O$ are dominators of $X$. (Note that the definition of dominator is essentially different in \cite{Gibson2010}).
	\label{locality_condition_def} 

\end{lemma}

%We first show that how locality condition leads the following lemma. Next we show that how a graph is  constructed that satisfies the locality condition. 

\begin{lemma}[\cite{Gibson2010}] \label{lemma_ptas_dds}
If the set of disks  $\mathcal{L} \cup \mathcal{O}$ satisfies the locality condition then  $|\mathcal{L}| \leq (1 + \epsilon) |\mathcal{O}|$ for $\epsilon = O(\frac{1}{\sqrt{t}})$.
\end{lemma}

\begin{proof}
We first note that the proof is in similar lines to the proof of Theorem \ref{ptas_dis_disks} (also similar to \cite{Chan2012}, \cite{Mustafa2010}, and \cite{Gibson2010}). Further, we adopt the terminologies used in the proof of Theorem \ref{ptas_dis_disks} and Lemma \ref{planar_sepa_lemma}. 

Assume that the set of disks  $\mathcal{L} \cup \mathcal{O}$ satisfies the locality condition.
%given in Lemma \ref{locality_condition_def}.
Further, let $G = (V, E)$ be any planar graph as in the locality condition with $V = \mathcal{O} \cup \mathcal{L}$. Let $r = t /(c_2+c_3)$ (note that $t$ is the local search parameter). We note that $|V_i \cup N(V_i)| \leq c_2 r + c_3 \sqrt{r} \leq t$. Since the \mdds~problem is a minimization problem and $\mathcal{L}$ is local optimum, we note that $|\mathcal{L}_i| \leq |\mathcal{O}_i| + |N(V_i)|$, otherwise an improvement is possible for the local solution $\mathcal{L}$ by replacing the disks in $\mathcal{L}_i$ with disks in $\mathcal{O}_i \cup N(V_i)$. By using the similar arguments in Theorem \ref{ptas_dis_disks}, we can obtain  $|\mathcal{L}| \leq \bigl(1 + O(1/\sqrt{t}) \bigr) |\mathcal{O}|$.\end{proof}

\subsection{Analysis of the algorithm}

In the following, we construct a graph $G = (V, E)$ which satisfies the locality condition given in Lemma \ref{locality_condition_def}.  Our construction of $G$ is inspired by the results in \cite{Madireddy2016}. Partition the set $\mathcal{R}$ into two sets $\mathcal{R}_1$ and $\mathcal{R}_2$ as follows:  
\begin{enumerate}
	\item $\mathcal{R}_1$ is the collection of disks in $\mathcal{R}$ such that for every disk  $D \in \mathcal{R}_1$ there exists  at least one point $p\in \mathcal{P}$ that is  covered by  $D$ and is also covered by  at least one disk in  $\mathcal{L}$ as well as  at least one disk in $\mathcal{O}$. 
	\item $\mathcal{R}_2 = \mathcal{R} \setminus \mathcal{R}_1$. 
\end{enumerate}

For every disk  $X \in \mathcal{L} \cup \mathcal{O}$, we consider a vertex in the graph $G$ at $\mathsf{cen}(X)$. The edge set $E$ is constructed in two phases, i.e., $E = E_1 \cup E_2$. The edge-sets $E_i$ (for $i=1 \text{ and } 2$) ensure that the locality condition is satisfied for the disks in $\mathcal{R}_i$. For a set of points $\{a, x_1, x_2, \ldots, x_k, b\}$, with $k \geq 1$ is an integer, the curve $\mathscr{C}(a, x_1, x_2, \ldots, x_k, b)$ connecting the points $a$ and $b$ is the chain of the segments $\seg(a,x_1), \seg(x_1, x_2), \seg(x_2, x_3), \ldots, \seg(x_{k-1}, x_k), \seg(x_k, b)$ such that no two of $\linebreak$ them  intersect except at the endpoints. 

\subsubsection{Phase I: Construction of the edge set $E_1$}

We first construct the $\mathsf{AWVD}$ (Additive Weighted Voronoi Diagram)  of the disks in $\mathcal{L} \cup \mathcal{O}$ as in Section \ref{ptas_mdis}.  For every disk $L \in \mathcal{L}$ and every disk $O \in \mathcal{O}$, we place an edge in $E_1$ with endpoints  $\mathsf{cen}(L)$ and $\mathsf{cen}(O)$ if and only if there exists a point $q$ in the plane, but not necessarily from $P$,  such that  $q$ is on the boundary of both $\mathsf{cell}(L)$ and $\mathsf{cell}(O)$. In particular, the edge is the curve $\mathscr{C}(\mathsf{cen}(L), q, \mathsf{cen}(O))$  (see Fig.  \ref{edgsetconstr_1} for an illustration).

\begin{figure}[ht!]
 	    \begin{center}
 	    \includegraphics[scale=01]{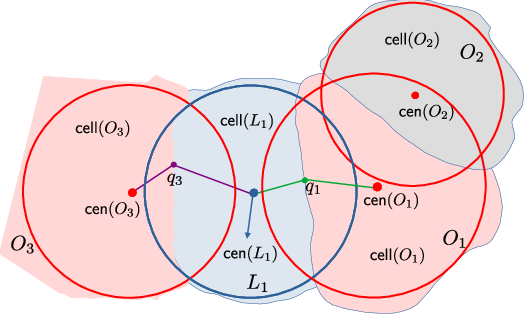}
 	    \caption{Here $L_1 \in \mathcal{L}$ and $O_1, O_2, O_3 \in \mathcal{O}$. Two edges $\mathscr{C}(\mathsf{cen}(L_1), q_1, \mathsf{cen}(O_1))$ (in green) and   $\mathscr{C}(\mathsf{cen}(L_1), q_3, \mathsf{cen}(O_3))$ (in purple) 	    added to set $E_1$.  }
 	    \label{edgsetconstr_1}    
 	    \end{center}
 	    
 	\end{figure}

\begin{lemma}\label{E_1_proof}
	The graph $G = (V, E_1)$ satisfies the locality condition for the disks in $\mathcal{R}_1$. 
\end{lemma}

\begin{proof}
Let $D$ be a disk in $\mathcal{R}_1$. Then there exists a point $p \in \mathcal{P}$ that is covered by $D$ as well as at least one disk in $\mathcal{L}$ and at least one disk in  $\mathcal{O}$. We consider that $p \in \mathsf{cell}(O)$ for some $O \in \mathcal{O}$ (when $p \in \mathsf{cell}(L)$ for some $L \in \mathcal{L}$, a similar argument can be given). Therefore, $\mathsf{\Phi}(O, p) \leq \mathsf{\Phi}(O^\prime, p) $ for all $O^\prime \in \mathcal{L} \cup  \mathcal{O}$.  Similarly, let $L \in \mathcal{L}$ be a disk such that $\mathsf{\Phi}(L, p) \leq \mathsf{\Phi}(L^\prime, p)$ for all $L^\prime \in \mathcal{L}$. Observe that disks  $O$ and  $L$  cover the point $p$. Otherwise, it contradicts that $D \in \mathcal{R}_1$. Now, there are two possible cases:

	\begin{description}
		\item[Case (a)] Suppose that $\mathsf{\Phi}(O, p) = \mathsf{\Phi}(L, p)$. 
		%\item[Case (a)] \textbf{[\boldmath{ $\mathsf{\Phi}(O, p) = \mathsf{\Phi}(L, p)$}]:}
		Then, both $\mathsf{cell}(O)$ and $\mathsf{cell}(L)$ share a common point $p$ and hence  there exists an edge between $\mathsf{cen}(D)$ and $\mathsf{cen}(L)$ in $E_1$.
		
		\item [Case (b)] 
		%\textbf{[\boldmath{$\mathsf{\Phi}(O, p) < \mathsf{\Phi}(L, p)$}]:}
		Suppose that $\mathsf{\Phi}(O, p) < \mathsf{\Phi}(L, p)$.  Consider the segment $\seg(\mathsf{cen}(L),  p)$. Note that  $\seg(\mathsf{cen}(L),  p)$ may contain points belonging to several cells.  Let $x \in \seg(\mathsf{cen}(L), p)$ be a point on the boundary of $\mathsf{cell}(L)$. Let $D^\prime \in \mathcal{L} \cup \mathcal{O}$ be a disk such that $\mathsf{\Phi}(D^\prime, x) = \mathsf{\Phi}(L, x)$. We now show that $D^\prime$ is a disk in $\mathcal{O}$ and covers $p$. Recall the assumption that no three points in $\mathcal{P}$ and disk centers are collinear. Then we have,
		 $||\mathsf{cen}(D^\prime) -  p|| < ||p -  x|| + ||\mathsf{cen}(D^\prime) -  x||$, that implies $\mathsf{\Phi}(D^\prime, p) < ||p -  x|| + \mathsf{\Phi}(D^\prime,  x)$  $ = ||p -  x|| + \mathsf{\Phi}(L, x) = \mathsf{\Phi}(L, p)$. By Lemma \ref{properties_WVD}, we conclude that the disk $D^\prime$ covers the point $p$. Further, $D^\prime$ is a disk in $\mathcal{O}$, as otherwise, the choice of $L$ is wrong. Note that, both disks $D^\prime$ and $L$ are the dominators of $D$ and $\mathsf{\Phi}(D^\prime, x) = \mathsf{\Phi}(L, x)$. Thus, there exists an edge  between $\mathsf{cen}(D^\prime)$ and $\mathsf{cen}(L)$ in $E_1$. 
	\end{description} 
Therefore, the lemma is proved.  \end{proof}

\subsubsection{Phase II: Construction of the edge set $E_2$}
%\noindent \textbf{{Phase II (construction of the edge set \boldmath{$E_2$}):}}   
For every disk $D \in \mathcal{R}$, let $\mathcal{P}^D \subseteq \mathcal{P}$ be the set of points that are covered by the disk $D$. Let $\mathcal{O}_{D} \subseteq \mathcal{O}$ and $\mathcal{L}_D \subseteq \mathcal{L}$  be the sets of dominators of $D$.  Further, let  $\mathcal{P}_\mathcal{L}^D \subseteq \mathcal{P}^D$ be the set of points covered by at least one disk in $\mathcal{L}_D$ and $\mathcal{P}_\mathcal{O}^D \subseteq \mathcal{P}^D$ be the set of points covered by at least one disk in $\mathcal{O}_D$. Note that  sets  $\mathcal{P}_\mathcal{L}^D$ and $ \mathcal{P}_\mathcal{O}^D$ are non-empty and disjoint, i.e.,  $\mathcal{P}_\mathcal{L}^D \cap \mathcal{P}_\mathcal{O}^D = \emptyset$. Further, $\mathcal{L}_D \cap \mathcal{O}_D = \emptyset$.  

\begin{lemma} \label{intersection_of_linesegments}
Let $p_1 \in \mathcal{P}_\mathcal{L}^D$ and $p_2 \in \mathcal{P}_\mathcal{O}^D$ be two points for some $D \in \mathcal{R}_2$. Further, let $p_1^\prime  \in \mathcal{P}_\mathcal{L}^{D^\prime}$ and $p_2^\prime \in \mathcal{P}_\mathcal{O}^{D^\prime}$ be the two points such that $p_1 \neq p_1^\prime $ and $p_2 \neq p_2^\prime $ for some $D^\prime \in \mathcal{R}_2$.  If the segments $\seg(p_1, p_2)$ and $\seg(p_1^\prime,  p_2^\prime)$ intersect then $\{p_1,p_2\}\cap \mathcal{P}^{D^\prime}\neq \emptyset$ or  $\{p_1^\prime,p_2^\prime\}\cap \mathcal{P}^{D}\neq \emptyset$. 
%at least one of the following is true: (i) $p_1$ or $p_2$ is covered by disk $D^\prime$ and (ii) $p_1^\prime$ or $p_2^\prime$ is covered by disk $D$. 
\end{lemma}

\begin{proof}
  We note that the segment $\seg(p_1, p_2)$ is completely inside $D$ and $\seg(p_1^\prime,  p_2^\prime)$ is completely inside $D^\prime$. If $p_1$ and $p_2$ are not covered by $D^\prime$, then $\seg(p_1,  p_2)$ intersects the disk's boundary $D^\prime$ exactly two times. Similarly, if both $p_1^\prime $ and $p_2^\prime$ are not covered by the disk $D$, then $\seg(p_1^\prime,  p_2^\prime)$ intersects the boundary of the disk $D$ twice. Note that no point in $\mathcal{P}$  lies on the boundary of any disk in $\mathcal{R}$.  Thus, if the statement of the lemma is not true, then boundaries of both $D$ and $D^\prime$ intersect four times, which is impossible.   	 
\end{proof}

\begin{lemma} \label{segment_edge}
   Let  $D  \in  \mathcal{R}_2$ be a disk. Further, let $p_1 \in \mathcal{P}_\mathcal{L}^D$ and $p_2 \in \mathcal{P}_\mathcal{O}^D$ be two points covered by $D$. Then, there exists a  segment $\seg(x_1, x_2) \subseteq \seg(p_1, p_2)$ such that 
   \begin{enumerate} 
       \item $x_1$ is a boundary point of $\mathsf{cell}(L)$ for some $L \in \mathcal{L}_D$, 
       \item $x_2$ is a boundary point of $\mathsf{cell}(O)$ for some $O \in \mathcal{O}_D$, and 
       \item there exists no other dominator,  $X \in \mathcal{L}_D \cup \mathcal{O}_D$, of $D$ such that $\mathsf{cell}(X)$ covers a point from $\seg(x_1, x_2)$. 
   \end{enumerate}

  \label{exist_non_dominator}
\end{lemma}

\begin{proof}

Let $\seg(x_1, x_2)$ be a minimal portion of $\seg(p_1, p_2)$ such that $x_1 \in \mathsf{cell}(L)$ and $x_2 \in \mathsf{cell}(O)$ for some disks $L \in \mathcal{L}_D$ and $O \in \mathcal{O}_D$. We note that such $L$ and $O$ exist because $p_1 \in \mathcal{P}_\mathcal{L}^D$ and $p_2 \in \mathcal{P}_\mathcal{O}^D$. Suppose that there exists a disk $X \in \mathcal{L}_D$ (the case $X \in  \mathcal{O}_D$ is similar) that covers a point $x \in \seg(x_1, x_2)$. Then the choice of $\seg(x_1, x_2)$ is wrong since $\seg(x, x_2)$ is minimal than $\seg(x_1, x_2)$ and satisfies the first two conditions of the lemma. Hence, no such disk $X \in   \mathcal{L}_D$ exists. \hfill  
\end{proof}

For a disk $D \in \mathcal{R}_2$, we call a segment $\seg(x_1, x_2) \subseteq \seg(p_1, p_2)$, where $p_1 \in \mathcal{P}_\mathcal{L}^D$ and $p_2 \in \mathcal{P}_\mathcal{O}^D$, that satisfy the conditions in Lemma \ref{segment_edge} as \textit{\colb{edge-segment}} of $D$. Note that, in general, the segment  $\seg(x_1, x_2)$ connects the boundaries of two cells in  $\mathsf{AWVD}$; one cell corresponding to a disk in $\mathcal{L}$ and other cell is corresponding to a disk in $\mathcal{O}$. 

Let $\mathcal{S}$ be the set of all possible non-overlapping edge segments such that for each disk $D$ in $\mathcal{R}_2$ there exists an edge segment $\seg(x_1, x_2) \subseteq \seg(p_1, p_2)$ where $p_1 \in \mathcal{P}_\mathcal{L}^D$ and $p_2 \in \mathcal{P}_\mathcal{O}^D$. We note that such $\mathcal{S}$ exists due to Lemma \ref{intersection_of_linesegments}.   In the following, we describe the construction of the edge set $E_2$: 

\begin{description}
    \item [Step 1.] Let $\seg(x, x^\prime)$ be a segment in $\mathcal{S}$ such that $\seg(x, x^\prime) \subseteq \seg(p, p^\prime)$ where $p \in \mathcal{P}_\mathcal{L}^D$ and $p^\prime \in \mathcal{P}_\mathcal{O}^D$ for some disk $D \in \mathcal{R}_2$.

    \item [Step 2.] Let  $L \in \mathcal{L}_D$ and $O \in \mathcal{O}_D$ be the two disks such that $x \in \mathsf{cell}(L)$ and $x^\prime  \in \mathsf{cell}(O)$.  

    \item [Step 3.] Place an edge $\mathscr{C}(\mathsf{cen}(L), x, x^\prime, \mathsf{cen}(O))$ 	in  $E_2$. The segment $\seg(\mathsf{cen}(L), x)$ is completely inside $\mathsf{cell}(L)$, the  segment $\seg(x^\prime,  \mathsf{cen}(O))$ is completely inside $\mathsf{cell}(O)$, and  the segment $\seg(x, x^\prime)$  is completely inside the disk $D$. See Fig.  \ref{fig:step3} for an illustration. 
\end{description}

\begin{figure}[h!]
    \centering
    \includegraphics[scale=1.25]{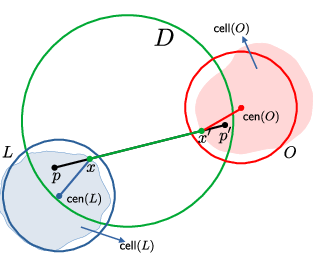}
    \caption{Illustration of segment $\seg(x, x^\prime)$ (in green). The edge $\mathscr{C}(\mathsf{cen}(L), x, x^\prime, \mathsf{cen}(O))$ is the chain of $\seg(\mathsf{cen}(L), x)$ (in blue),  $\seg(x, x^\prime)$ (in green), and $\seg(x^\prime,  \mathsf{cen}(O))$ (in red). }
    \label{fig:step3}
\end{figure}

\subsubsection{Proving the planarity of $G$}

We note that in the graph $G= (V, E_1 \cup E_2)$, for every disk $D \in \mathcal{R}$, there exists an edge between a dominator of $D$ in $\mathcal{L}$ and a dominator of $D$ in  $\mathcal{O}$. No two edges in $E_1$ intersect, but an edge in $E_1$ and an edge in $E_2$ may intersect, or two edges in $E_2$ may intersect.  Thus, the graph $G = (V, E_1 \cup E_2)$ may not be planar. Hence, the locality condition, given in  Lemma  \ref{locality_condition_def},  may not be satisfied.  In the following, we show that one can obtain a planar graph by  \textit{\colb{edge perturbation}}  of some edges without violating the locality condition for the disks in $\mathcal{R}$.  For the sake of ease of notation, we assume that each edge in $E = E_1 \cup E_2$ is of the form $\mathscr{C}(\mathsf{cen}(L), x, x^\prime, \mathsf{cen(O)})$ for some $L \in \mathcal{L}, O \in \mathcal{O}$, with $x$ and $x^\prime$ are points in the plane. In particular, if $\mathscr{C}(\mathsf{cen}(L), x, x^\prime, \mathsf{cen(O)}) \in E_1$, then $x$ and $x^\prime$ are the same points, i.e., $x=x^\prime$.  
 
 Let  %$e_1 = \overline{\mathsf{cen}(L_1)x_1} \cup \overline{x_1 x_1^\prime} \cup \overline{\mathsf{cen}(O_1)x_1^\prime}$ and $e_2 = \overline{\mathsf{cen}(L_2)x_2} \cup \overline{x_2x_2^\prime} \cup \overline{\mathsf{cen}(O_2)x_2^\prime}$ 
 $e_1 = \mathscr{C}(\mathsf{cen}(L_1), x_1, x_1^\prime,  \mathsf{cen}(O_1))$ and $e_2 = \mathscr{C}(\mathsf{cen}(L_2), x_2, x_2^\prime, \mathsf{cen}(O_2))$
 be the two edges in $E_1 \cup E_2$ for some $L_1, L_2 \in \mathcal{L}$ and $O_1, O_2 \in \mathcal{O}$ (see Fig. \ref{edge_intersect}). Assume that   $L_1$ and $O_1$ are dominators of a disk $D_1 \in \mathcal{R}_2$ and  $L_2$ and $O_2$ are dominators of a disk $D_2 \in \mathcal{R}_2$, i.e., both $e_1, e_2 \in E_2$ (the case where either $e_1 \in E_1$ or $e_2 \in E_1$ is similar, see Figure \ref{edge_intersect_b} for a pictorial evidence). Further, assume that $\seg(x_1, x_1^\prime)$ is a portion of $\seg(p_1, p_1^\prime)$ and $\seg(x_2, x_2^\prime)$ is a portion of $\seg(p_2, p_2^\prime)$ for some points $p_1, p_2, p_1^\prime, p_2^\prime \in \mathcal{P}$ such that  $p_1 \in \mathcal{P}^{L_1} \cap \mathcal{P}^{D_1}$, $p_1^\prime \in \mathcal{P}^{O_1} \cap \mathcal{P}^{D_1} $,  $p_2 \in \mathcal{P}^{L_2} \cap \mathcal{P}^{D_2}$, and $p_2^\prime  \in \mathcal{P}^{D_2} \cap \mathcal{P}^{O_2}$.

 \begin{figure}[ht!]
\begin{center}
{\subfigure[ ]{\includegraphics[scale=1.25]{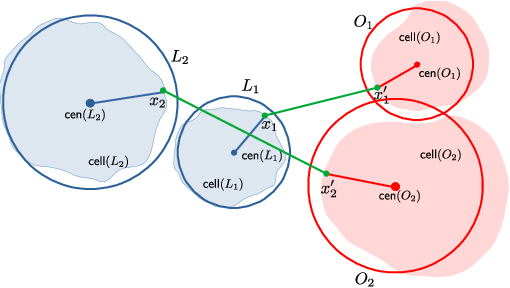}
\label{edge_intersect_a}
}}
{\subfigure[ ]{\includegraphics[scale=1.25]{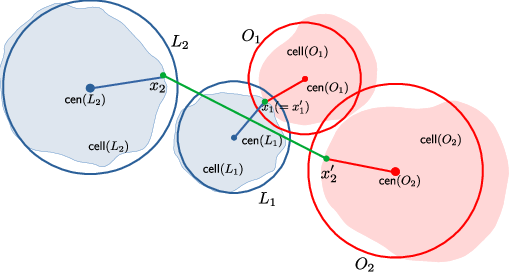}
\label{edge_intersect_b}
}}
\end{center}
\caption{A possible placement of disks and points such that edges $e_1 = \mathscr{C}(\mathsf{cen}(L_1), x_1, x_1^\prime,  \mathsf{cen}(O_1))$ and $e_2 = \mathscr{C}(\mathsf{cen}(L_2), x_2, x_2^\prime, \mathsf{cen}(O_2))$ in $E = E_1 \cup E_2$ intersect. \subref{edge_intersect_a} Both edges $e_1, e_2 \in E_2$  \subref{edge_intersect_b} $e_1 \in E_1$ and $e_2 \in E_2$. }
\label{edge_intersect}
\end{figure}

 We note that the segments  $\seg(x_1, x_1^\prime)$ and $\seg(x_2, x_2^\prime)$ also do not intersect otherwise $\seg(p_1, p_1^\prime)$ and $\seg(p_2, p_2^\prime)$ intersect, this contradicts the fact that no two segments in $\mathcal{S}$ intersect. 
 
 \begin{lemma} 
  The following pairs of segments do not intersect: (i) $\seg(\mathsf{cen}(L_1), x_1)$ and $\seg(\mathsf{cen}(L_2), x_2)$, (ii)
 $\seg(\mathsf{cen}(L_1), x_1)$ and $\seg(\mathsf{cen}(O_2), x_2^\prime)$, (iii) $\seg(\mathsf{cen}(O_1), x_1^\prime)$ and $\seg(\mathsf{cen}(L_2), x_2)$, and (iv)
 $\seg(\mathsf{cen}(O_1), x_1^\prime)$ and $\seg(\mathsf{cen}(O_2), x_2^\prime)$.
 \end{lemma}
\begin{proof}
	Suppose $\seg(\mathsf{cen}(L_1), x_1)$ and $\seg(\mathsf{cen}(L_2), x_2)$ intersect at a point $x$. Since $\seg(\mathsf{cen}(L_1), x_1)$ is completely inside $\mathsf{cell}(L_1)$ and  $\seg(\mathsf{cen}(L_2), x_2)$ is completely inside $\mathsf{cell}(L_2)$, the point $x$ cannot be an interior point to both  $\mathsf{cell}(L_1)$ and $\mathsf{cell}(L_2)$. Hence, $x=x_1=x_2$. Thus, both segments $\seg(p_1, p_1^\prime)$ and $\seg(p_2, p_2^\prime)$ intersect at $x$, which is not the common endpoint of the two segments. This contradicts that no two segments in $\mathcal{S}$ intersect. 
	
	The other three cases are similar. \hfill     
\end{proof}

Thus, if both the edges $e_1$ and $e_2$ intersect, then it must be the case that exactly one of the following four pairs of segments intersects. $(i)$ $\seg(\mathsf{cen}(L_1), x_1)$ and $\seg(x_2, x_2^\prime)$, $(ii)$ $\seg(\mathsf{cen}(O_1), x_1^\prime)$ and $\seg(x_2, x_2^\prime)$, $(iii)$ $\seg(\mathsf{cen}(L_2), x_2)$ and $\seg(x_1, x_1^\prime)$, and $(iv)$ $\seg(\mathsf{cen}(O_2), x_2^\prime)$ and $\seg(x_1, x_1^\prime)$. We now describe the edge perturbation step for the first case. The other three cases are similar. 

 Suppose that $\seg(\mathsf{cen}(L_1), x_1)$ and $\seg(x_2, x_2^\prime)$  intersects at a point $x \in \mathsf{cell}(L_1)$. Hence,   $\seg(\mathsf{cen}(L_1), x_1)$ and $\seg(p_2, p_2^\prime)$ also intersects at the same point $x$. Further, the boundary of the disk $L_1$ intersects $\seg(p_2, p_2^\prime)$.    Otherwise, $L_1$ covers both $p_2$ and $p_2^\prime$. Thus, $p_2^\prime \in L_1 \cap O_2$. This contradicts that $D_2 \in \mathcal{R}_2$. In particular, $L_1$ cannot cover $p_2^\prime$.

 \begin{lemma}
     The segment $\seg(p_2, p_2^\prime)$  intersects the boundary of $L_1$ exactly two times. \label{twice_intersect}
 \end{lemma}
 \begin{proof}
From the above discussion, it is clear that $L_1$ does not cover $p_2$. For the sake of contradiction, assume that $L_1$ covers $p_2^\prime$. Since point $x \in \mathsf{cell}(L_1)$ and $L_1$ is a dominator of $D_2$, the choice of $\seg(x_2, x_2^\prime$) is wrong. Hence, $L_1$ does not cover $p_2^\prime$. Therefore, $L_1$ intersect $\seg(p_2, p_2^\prime)$ exactly twice. 
 %For the sake of contradiction, assume that  $\seg(p_2, p_2^\prime)$ intersects the boundary of the disk $L_1$ exactly once. Hence,   $p_2$ is inside $L_1$; thus, $L_1$ is also a dominator of $D_2$. Since $x \in \mathsf{cell}(L_1)$ and $x$ lies on $\seg(p_2, p_2^\prime)$,  from the construction of edges in $E_2$, it must be the case that either $x_2 = x$ or $x_2$ is an interior point on $\seg(x, x_2^\prime)$. In the latter case,   $\seg(\mathsf{cen}(L_1), x_1)$ and $\seg(x_2, x_2^\prime)$ do not intersect, which is a contradiction. Hence, $x_2 =x$.  Note that $\seg(x, x_2^\prime)$ is the minimal portion of $\seg(p_2, p_2^\prime)$ such that $x \in \mathsf{cell}(L_1)$  and $x_2^\prime \in \mathsf{cell}(O_2)$. Further,  $x$ must be a boundary point of $\mathsf{cell}(L_1)$. Hence, $\mathsf{\Phi}(L_1, x) = \mathsf{\Phi}(X, x)$, for some $X \in \mathcal{L} \cup \mathcal{O}$. 
 %Consider$||\mathsf{cen}(X) - x_1|| < ||x_1 -  x|| + ||\mathsf{cen}(X) - x||$ which implies that 
%				$\mathsf{\Phi}(X, x_1) < ||x_1 - x|| + \mathsf{\Phi}(X,  x)$  $ = ||x_1 - x|| + \mathsf{\Phi}(L_1, x) = \mathsf{\Phi}(L_1, x_1)$. Hence, $x_1 \notin \mathsf{cell}(L_1)$, which is a contradiction. 		
%Hence, the lemma is true.  \hfill   
 \end{proof}
 
 \noindent
\textbf{Edge perturbation step: }
 
 From Lemma \ref{twice_intersect}, it is clear that  $\seg(p_2, p_2^\prime)$ intersects the disk $L_1$ twice. Rotate the plane such that a horizontal line $l$ passes through $\seg(p_2, p_2^\prime)$. Without loss of generality, assume that $\mathsf{cen}(L_1)$ is above the line $l$. Hence, $x_1$ lies below $l$. Partition the disk $L_1$ into two connected regions that are on both sides of $\seg(p_2, p_2^\prime)$; one is above $l$ (call the region $L_1^+$) and the other one is below $l$ (call the region $L_1^-$). The point $p_1$ lies inside the region $L_1^-$, otherwise $\seg(p_1, p_1^\prime)$ and $\seg(p_2, p_2^\prime)$ intersects.  Recall that $\seg(p_2, p_2^\prime)$ completely lies inside the disk $D_2$. 
 Clearly, exactly one of the regions $L_1^+$ and $L_1^-$ falls completely inside $D_2$.  
 
 If $L_1^-$ is completely covered by $D_2$, then $L_1$ is a dominator of $D_2$, this contradicts the choice of $\seg(x_2, x_2^\prime)$.  Hence, $L_1^+$ is completely inside $D_2$. Further, there exists a point $q^\prime \in \mathcal{P}^{L_1} \setminus \mathcal{P}^{D_2}$,  which is inside the region $L_1^-$ otherwise $\mathcal{P}^{L_1} \subset \mathcal{P}^{D_2}$, this is a contradiction. 
 Further, note that the region $L_1^+$ does not contain any point from $\mathcal{P}^{L_1}$.

We now slightly modify the portion  $\seg(x_2, x_2^\prime)$ of the edge $e_2$ as follows: let $\seg(t_2, t_2^\prime)$ be the maximal portion of $\seg(x_2, x_2^\prime)$ such that $\seg(t_2, t_2^\prime)$ is completely in $\mathsf{cell}(L_1)$ (see Fig. \ref{edge_per_1}) and $x$ is an interior point to $\seg(t_2, t_2^\prime)$ due to the fact that $x$ is not the boundary point of $\mathsf{cell}(L_1)$. We know that $\mathsf{cen}(L_1)$ is not a boundary point for $\mathsf{cell}(L_1)$. Let $t$ be the point on the boundary of $\mathsf{cell}(L_1)$ with same $x$-coordinate as $\mathsf{cen}(L_1)$ and the $y$-coordinate of $t$ is greater than the $y$-coordinate of $\mathsf{cen}(L_1)$ (see Fig.  \ref{edge_per_1}). Note that $\triangle t_2 t t_2^\prime $ lies completely inside $\mathsf{cell}(L_1)$. We further note that some of the edges in $E$ may intersect the segment $\seg(\mathsf{cen}(L_1), t)$. Let $t^*$ be the first point on the segment $\seg(\mathsf{cen}(L_1), t)$ at which some edge in $E$ intersects $\seg(\mathsf{cen}(L_1), t)$ when we walk from $\mathsf{cen}(L_1)$ to $t$ along the segment $\seg(\mathsf{cen}(L_1), t)$.  If no edge in $E$ intersects $\seg(\mathsf{cen}(L_1), t)$ then set $t^* = t$. Note that $t^*$ cannot be $\mathsf{cen}(L_1)$ since we assume that no three points from the set of disk centers union the set of points in $\mathcal{P}$ are collinear. We replace $\seg(t_2, t_2^\prime)$ with a non-self-intersecting curve $\mathcal{C}(t_2, t_2^\prime)$ such that all the points on this curve are inside $\triangle t_2 t t_2^\prime$ and passes through a point on $\seg(\mathsf{cen}(L_1), t^*)$  (see Fig. \ref{edge_per_1}). Suppose that several segments of the form $\seg(\mathsf{cen}(L_1), x_j)$ intersect $\seg(x_2, x_2^\prime)$. We note that  $\seg(t_2, t_2^\prime)$ is still be the maximal portion of $\seg(x_2, x_2^\prime)$ such that $\seg(t_2, t_2^\prime)$ is completely contained in $\mathsf{cell}(L_1)$. Hence, in this case also,  $\seg(t_2, t_2^\prime)$ will be replaced with $\mathcal{C}(t_2, t_2^\prime)$. 
   A non-self-intersecting curve $\mathscr{C}_1(x_2, x_2^\prime)$ connecting the points $x_2$ and $x_2^\prime$ is the chain of  $\seg(x_2, t_2)$,  $\mathcal{C}(t_2, t_2^\prime)$, and  $\seg(t_2^\prime,  x_2^\prime)$.    Now the segment $\seg(x_2, x_2^\prime)$ in the edge $e_2 = \mathscr{C}(\mathsf{cen}(L_2), x_2, x_2^\prime, \mathsf{cen}(O_2))$ is replaced with $\mathscr{C}_1(x_2, x_2^\prime)$.

  	\begin{figure}[ht!]
 	    \begin{center}
 	    \includegraphics[scale=1.5]{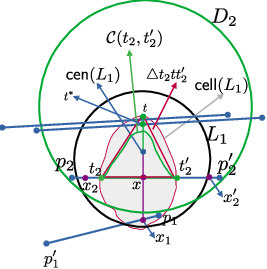}
 	    \caption{Edge perturbation step illustration. The shaded region defines $\mathsf{cell}(L_1)$. The segment $\seg(x_2, x_2^\prime)$ in the edge  is replaced with the curve $\mathscr{C}_1(x_2, x_2^\prime)$ (the chain of segment segment $\seg(x_2, t_2)$, curve $\mathcal{C}(t_2, t_2^\prime)$, and segment $\seg(t_2^\prime,  x_2^\prime)$) which  connects $x_2$ and $x_2^\prime$.  }
 	    \label{edge_per_1}    
 	    \end{center}
 	    
 	\end{figure}

Suppose that the segments $\seg(\mathsf{cen}(X_i), x_i), \seg(\mathsf{cen}(X_{i+1}), x_{i+1}), \ldots,  \seg(\mathsf{cen}(X_j), x_j) $ intersect $\seg(x_2, x_2^\prime)$ in the order from left to right (when we walk from $x_2$ to $x_2^\prime$ along $\seg(x_2, x_2^\prime)$) where $X_i, \ldots, X_j \in (\mathcal{L} \cup \mathcal{O}) \setminus \{L_2, O_2\}$.  In this case, we simultaneously apply the above edge perturbation step for all disks $X_i, X_{i+1}, X_j$. Note that no two curves introduced at this step intersect since each curve completely lies inside a different cell in the additive weighted Voronoi diagram.   We now give the general structure of the edge $e_2$. For $\lambda = i, i+1, \ldots, j$, let $\seg(t_\lambda,  t_\lambda^\prime)$ be the maximal portion of $\seg(x_2, x_2^\prime)$ such that $\seg(t_\lambda, t_\lambda^\prime)$ is completely inside $\mathsf{cell}(X_\lambda)$. Define a non-self-intersecting curve $\mathcal{C}(t_\lambda , t_\lambda ^\prime)$ as before. Define a curve $\mathscr{C}_1^*(x_2, x_2^\prime)$ as the chain of  $\seg(x_2, t_i)$,   $\mathcal{C}(t_i, t_i^\prime)$,  $\seg(t_i^\prime, t_{i+1})$,  $\mathcal{C}(t_{i+1}, t_{i+1}^\prime), \ldots, $  $\mathcal{C}(t_j, t_j^\prime)$, and  $\seg(t_j^\prime,  x_2^\prime)$. In the edge $e_2 = \mathscr{C}(\mathsf{cen}(L_2), x_2, x_2^\prime, \mathsf{cen}(O_2))$, the segment $\seg(x_2, x_2^\prime)$ is replaced with $\mathscr{C}_1^*(x_2, x_2^\prime)$. In particular,  the edge $e_2$ is the curve $\mathscr{C}^*(\mathsf{cen}(L_2), x_2, x_2^\prime, \mathsf{cen}(O_2)) $,  joining $\mathsf{cen}(L_2)$ and $\mathsf{cen}(O_2)$, which is the chain of  $\seg(\mathsf{cen}(L_2), x_2)$,  $\mathscr{C}_1^*(x_2, x_2^\prime)$, and  $\seg(x_2^\prime,  \mathsf{cen}(O_2))$. 
   In general,  after applying edge perturbation step,  an edge $e =\mathscr{C}(\mathsf{cen}(L), x, x^\prime, \mathsf{cen(O)}) \in E$, will transform to a curve $\mathscr{C}^*(\mathsf{cen}(L), x, x^\prime, \mathsf{cen}(O)) $.

Suppose that $\seg(\mathsf{cell}(L_1), x_1)$ intersects many segments of the form $\seg(x_i, x_i^\prime) $ ($i \neq 1$). Let  $\seg(x_2,  x_2^\prime), \seg(x_3, x_3^\prime), \ldots, \seg(x_j, x_j^\prime)$ be the ordering of the segments which intersect $\seg(\mathsf{cell}(L_1), x_1)$ when we walk from $\mathsf{cell}(L_1)$ to $x_1$  along the segment $\seg(\mathsf{cen}(L_1), x_1)$. Let  $x^*_{i}$ be the intersection point of $\seg(x_i, x_i^\prime)$ and $\seg(\mathsf{cell}(L_1), x_1)$. As before,  let $\seg(t_i, t_i^\prime)$ be the maximal portion of $\seg(x_i, x_i^\prime)$ such that $\seg(t_i, t_i^\prime)$ is completely inside $\mathsf{cell}(L_1)$. Similar as above, define the point $t$ on the boundary of $\mathsf{cell}(L_1)$, then define $t^*$ on $\seg(\mathsf{cen}(L_1), t)$ such that if an edge in $E$ intersects the segment $\seg(\mathsf{cen}(L_1), t)$ then it intersects $\seg(\mathsf{cen}(L_1), t)$ at a point on the segment $\seg(t, t^*)$ (if no edge in $E$ intersect $\seg(\mathsf{cen}(L_1), t)$, then set $t^* = t$). As mentioned before, $t^*$ cannot be $\mathsf{cen}(L_1)$.   Let  $t_2^*, t_3^*, \ldots, t_j^*$ be some points on the segment $\seg(\mathsf{cen}(L_1), t)$ such that  when we walk from $t^*$ to $\mathsf{cen}(L_1)$ along the segment $\seg(t^*,  \mathsf{cen}(L_1))$, the points appear in the same order. We define a non-self-intersecting curve $\mathcal{C}(t_i,  t_i^\prime)$, as discussed before, such that all the points on the curve are inside $\triangle t_i t t_i^\prime$ and it passes through $t_i^*$  (see Fig. \ref{edge_per_4}). We note that one can draw all these curves $\mathcal{C}(t_2,  t_2^\prime)$, $\mathcal{C}(t_3,  t_3^\prime), \ldots, \mathcal{C}(t_2,  t_2^\prime)$, one after another in the same order,  such that no two of these curves intersect (see Fig. \ref{edge_per_4}).

 	\begin{figure}[ht!]
 	    \begin{center}
 	    \includegraphics[scale=1.5]{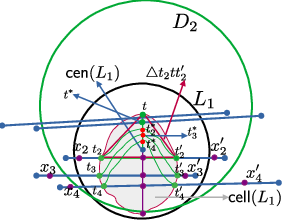}
 	    \caption{Edge perturbation step illustration when multiple edges intersect $\seg(\mathsf{cen}(L_1), x_1)$.  }
 	    \label{edge_per_4}    
 	    \end{center}
 	    
 	\end{figure}

We apply the edge perturbation method for all the pairs of edges $e$ and $e^\prime$ in $E_1 \cup E_2$ if $e$ and $e^\prime$ intersect.  Let $G$ be the resultant graph after applying the edge permutation on each pair of edges in $E_1 \cup E_2$ if needed. We note that no two curves introduced in this intersect due to the fact that, 
\begin{itemize} 
    \item We can draw all the curves inside the same cell without intersecting each other and no curve intersects other edges which pass through the points in the same cell and 
    \item The curves that belong to two different cells do not intersect since the curves are interior to the cell. 
\end{itemize}

Thus,  we have the following lemma. 

\begin{lemma} \label{lem:planar}
  The graph $G$ is a planar graph.
\end{lemma}

From the construction of the graph $G$ and from Lemma \ref{lem:planar}, we conclude the following lemma. 

\begin{lemma}\label{locality_condition_proof}
   The disks in $\mathcal{L} \cup \mathcal{O}$  satisfy the locality condition. 
\end{lemma}
Now we have the following theorem.
\begin{theorem}
The \mdds~problem with arbitrary radii disks admits a \ptas, i.e., For any integer $t > 1$, a $t$-level local search produces a solution $\mathcal{L}$ of size $\leq (1 + \epsilon) |\mathcal{O}|$,  $\epsilon = O(\frac{1}{\sqrt{t}})$, in $nm^{O(1/\epsilon^2)}$-time, where $\cal O$ is an optimum solution of the \mdds~problem, $m$ is the number of disks, and $n$ is the number of points. 
 \label{ptas_dds_disks}
\end{theorem}

\begin{proof}%{Proof of Theorem \ref{ptas_dds_disks}}
Lemma \ref{lemma_ptas_dds} together with Lemma \ref{locality_condition_proof} gives the proof of the theorem. \hfill   
\end{proof}

\begin{theorem}
    
The \mdds~problem with arbitrary side length axis-parallel squares admits \ptas, i.e., For any integer $t > 1$,  a $t$-level local search produces a solution $\mathcal{L}$ of size $\leq (1 + \epsilon) |\mathcal{O}|$,  $\epsilon = O(\frac{1}{\sqrt{t}})$, in $O(nm^{O(1/\epsilon^2)})$-time, where $\cal O$ is an optimum solution of the \mdds~problem, $m$ is the number of squares of arbitrary side length, and $n$ is the number of points.  %The analysis is similar to the analysis of arbitrary radii disks, except that for any two points $p$ and $q$ in the plane, $||p - q||$ is defined under infinity norm $L_\infty$ instead of $L_2$ norm as in \cite{Aschner2013}.

\label{ptas_dds_squares}

\end{theorem}
\begin{proof}    
Let $\mathcal{R}$ be the set of axis-parallel squares with arbitrary side lengths. Apply $t$-level local search given for the \mdds~problem for arbitrary radii disks with $t = O( 1/\epsilon^2)$ for the instance ($\mathcal{P}, \mathcal{R})$.  The analysis is similar to the analysis of arbitrary radii disks,  except that for any two points $p$ and $q$ in the plane, $\mathsf{dist}(p,q)$ is defined with respect to the infinity norm,  ${L}_\infty$,  instead of $L_2$-norm as in \cite{Aschner2013} (see proof of Theorem \ref{ptas_dis_squares} for the definition of the distance function and a cell). For the instance $(\mathcal{P}, \mathcal{R})$ of the \mdds~problem (where the objects in $\mathcal{R}$ are axis-parallel squares), let
 $\mathcal{L} \subseteq \mathcal{R}$ be the $t$-locally optimal  solution return by  the $t$-level local search and  let $\mathcal{O} \subseteq \mathcal{R}$ be an optimal solution. We note that Lemma \ref{lemma_ptas_dds} is still true i.e., $|\mathcal{L}| \leq (1 + \epsilon) |\mathcal{O}|$ for $\epsilon = O(\frac{1}{\sqrt{t}})$  if the squares in $\mathcal{L} \cup \mathcal{O}$ satisfy the locality condition given in Lemma \ref{locality_condition_def}. The proof of showing that the squares in $\mathcal{L} \cup \mathcal{O}$ satisfy the locality condition is similar to the case of arbitrary radii disks.    
\end{proof}

\section{\texorpdfstring{\boldmath{\apx}}~-hardness Results} \label{apx_hard_results}

In this section, we present \apx-hardness results for the \mdis~and \mdds~problems.    First,  we  define a restricted version of the \mds~problem with set systems, the \spds~problem, and show that it is \apx-hard. We use the \spds~to prove Theorem \ref{apx_hardness_mdds}. The work is inspired by the results in \cite{Chan2014}.

\begin{definition} [\boldmath{\spds}] Let $(\mathcal{U}, \mathcal{S})$ be a range space where $\mathcal{U} = \mathcal{A} \cup \mathcal{B}$, $\mathcal{A} = \{a_1, a_2, \ldots, a_n\}$, $\mathcal{B}  = \mathcal{B}^1 \cup \mathcal{B}^2 \cup \cdots \cup \mathcal{B}^6$, and  $\mathcal{B}^i = \{b_1^i, b_2^i, \ldots, b_m^i\}$ for $1 \leq i \leq 6$ such that $3m=  2n$. Further, $\mathcal{S}$ is a collection of $7m$ subsets of $\mathcal{U}$ such that
\begin{enumerate}
		
\item Every element in $\mathcal{U}$ is in exactly two sets in $\mathcal{S}$. 
		
\item For every $t$, ($1 \leq t \leq m$), there exist three integers $ 1 \leq i < j < k \leq n$ such that the sets $\{a_i, b_t^1\}$, $\{b_t^1, b_t^2\}$, $\{b_t^2, b_t^3\}$, $\{b_t^3, b_t^4, a_j\}$, $\{b_t^4, b_t^5\}$, $\{b_t^5, b_t^6\}$, and $\{b_t^6, a_k\}$ are in the collection $\mathcal{S}$. 
\end{enumerate}

The objective is to find a minimum size sub-collection $\mathcal{S}^\prime \subseteq \mathcal{S}$ such that for every $S \in \mathcal{S}$,  either $S \in \mathcal{S}^\prime$ or there exists a set $S^\prime \in \mathcal{S}^\prime$ such that $S \cap S^\prime \neq \emptyset$. 
	\label{special_3ds_def}
\end{definition}

We use the \colb{$L$-reduction} \cite{PAPADIMITRIOU1991}  to prove that the \spds~is \apx-hard. 
Let $X$ and $Y$ be two optimization problems. 
A polynomial-time computable function $f$ from $X$ to $Y$ is an $L$-reduction if there exist two positive constants $\alpha$ and $\beta$ (usually 1) such that  for each  instance $x$ of $X$ the following two conditions hold:

\begin{description}
    \item[C1:] $OPT(f(x)) \leq \alpha \cdot OPT(x)$ where $OPT(x)$ and $OPT(f(x))$ are the size of the optimal solutions of $x$ and $f(x)$,  respectively. 
    \item[C2:] For any given solution of  $f(x)$ with cost $C_{f(x)}$,  there exists a polynomial-time algorithm that finds a feasible solution of $x$ with cost $C_x$ such that $|C_x -OPT(x)| \leq \beta \cdot |C_{f(x)} - OPT(f(x))|$. 
\end{description}
 
\begin{lemma}\label{special_3ds_apx_hard}
\spds~is \apx-hard.  
\end{lemma}	

\begin{proof}
We prove the lemma by giving an $L$-reduction from an \apx-hard problem, dominating set on cubic graphs \cite{ALIMONTI}. Let $I_1$ be an instance of dominating set problem on a graph $G  = (V, E)$ with $V = \{v_1, v_2, \ldots, v_m\}$  and $E = \{e_1, e_2, \ldots, e_n\}$ such that the degree of every vertex in $V$ is exactly  three.  We now generate an instance $I_2$ of \spds~from $I_1$ as follows:
\begin{enumerate}
	
\item Let $\mathcal{A} = \{a_1, a_2, \ldots, a_n\}$ and $\mathcal{B} = \mathcal{B}^1 \cup \mathcal{B}^2 \cup \cdots \cup \mathcal{B}^6$ where $\mathcal{B}^i = \{b_1^i, b_2^i, \ldots, b_m^i\}$ for $i=1, 2, \ldots, 6$.
\item For a vertex   $v_t$    in $V$ ($1 \leq t \leq m$),  let $e_i, e_j,$ and $e_k$ ($1\leq i < j < k \leq n$) be the edges incident on $v_t$.  Then add seven sets $\{a_i, b_t^1\}$, $\{b_t^1, b_t^2\}$, $\{b_t^2, b_t^3\}$, $\{b_t^3, b_t^4, a_j\}$, $\{b_t^4, b_t^5\}$, $\{b_t^5, b_t^6\}$, and $\{b_t^6, a_k\}$ into $\mathcal{S}$. Do the same for every vertex in $V$. 
	
\end{enumerate}
	
Let $\mathcal{O}(I_1) \subseteq V$  be an optimal dominating set for the instance $I_1$. We now give a polynomial-time algorithm to find an optimal solution  $\mathcal{O}(I_2)$ for instance $I_2$ of the \spds~problem  from $\mathcal{O}(I_1)$. For every vertex $v_t \in V(G)$, do the following:
\begin{enumerate}
\item If $v_t$ is in $\mathcal{O}(I_1)$ then take the sets $\{a_i, b_t^1\}$, $\{b_t^3, a_j, b_t^4\}$, $\{b_t^6, a_k\}$ in $\mathcal{O}(I_2)$. 
		
\item If $v_t$ is not in $\mathcal{O}(I_1)$ then take the sets $\{b_t^2, b_t^3\}$, $\{b_t^4, b_t^5\}$ in $\mathcal{O}(I_2)$. 
\end{enumerate}
	
One can easily verify that  $\mathcal{O}(I_2)$ is an optimal dominating set for $I_2$ and $|\mathcal{O}(I_2)| = |\mathcal{O}(I_1)| + 2 m $.  Since $|\mathcal{O}(I_1)| \geq m/4$  we get $|\mathcal{O}(I_2)| \leq 9 \cdot  |\mathcal{O}(I_1)| $. 
	Similarly, for any given  feasible solution $\mathsf{F}_2 \subseteq \mathsf{S}$ of $I_2$, one can obtain a feasible solution  $\mathsf{F}_1 \subseteq V(G)$ of $I_1$ such that $|\mathsf{F}_1| \leq |\mathsf{F}_2| - 2m$. 
	
	Thus,  we conclude that the above reduction is an $L$-reduction \cite{PAPADIMITRIOU1991}  with $\alpha=9$ and $\beta=1$. Therefore, the \spds~problem is \apx-hard.  
 \end{proof}

\begin{theorem}\label{apx-hardness-theorem}
	The \mdds~problem is \apx-hard for the following classes of geometric objects.
	\begin{description} 
		
		\item[{\bf A1}]  Axis-parallel rectangles in $\mathbb{R}^2$, even when all rectangles have an upper-left corner inside a square with side length $\epsilon$ and lower-right corner inside a square with side length $\epsilon$ for an arbitrary small $\epsilon>0$.
		
		\item[{\bf A2}] Axis-parallel  ellipses in $\mathbb{R}^2$, even when all the ellipses  contain the origin. 
		
		\item[{\bf A3}] Axis-parallel strips in $\mathbb{R}^2$. 
		\item[{\bf A4}]   Axis-parallel rectangles in $\mathbb{R}^2$,  even when every pair of rectangles intersects either zero or four times. 
		\item[{\bf A5}] Downward shadows of  segments in the plane.

		\item[{\bf A6}]  Downward shadows of cubic polynomials in the plane. 
		
		\item[{\bf A7}] Unit ball in $\mathbb{R}^3$, even when the origin is inside every unit ball.
		
		\item[{\bf A8}] Axis-parallel cubes of similar size in $\mathbb{R}^3$ containing a common point.
		
		\item[{\bf A9}] Half-spaces in $\mathbb{R}^4$. 
		
		\item[{\bf A10}] Fat semi-infinite wedges in $\mathbb{R}^2$ with apices near the origin. 
	\end{description}
	\label{apx_hardness_mdds}
\end{theorem}

\begin{proof}
%{Proof of Theorem \ref{apx_hardness_mdds}.} 
The proof is essentially similar to the results in  \cite{Chan2014}. In the following, for a given instance of the \spds~problem, we give an encoding of the \mdds~problem for each class of objects. Let $(\mathcal{U}, \mathcal{S})$ be a range space where $\mathcal{U} = \mathcal{A} \cup \mathcal{B}$, $\mathcal{A} = \{a_1, a_2, \ldots, a_n\}$, and  $\mathcal{B}  = \mathcal{B}^1 \cup \mathcal{B}^2 \cup \cdots \cup \mathcal{B}^6$, where $\mathcal{B}^i = \{b_1^i, b_2^i, \ldots, b_m^i\}$, for $1 \leq i \leq 6$, with $3m=  2n$. Further, $\mathcal{S}$ is a collection of $7m$ subsets of $\mathcal{U}$ such that
\begin{enumerate}
		
\item Every element in $\mathcal{U}$ is in exactly two sets in $\mathcal{S}$. 
		
\item For every $t$, ($1 \leq t \leq m$), there exist three integers $ 1 \leq i < j < k \leq n$ such that the sets $\{a_i, b_t^1\}$, $\{b_t^1, b_t^2\}$, $\{b_t^2, b_t^3\}$, $\{b_t^3, b_t^4, a_j\}$, $\{b_t^4, b_t^5\}$, $\{b_t^5, b_t^6\}$, and $\{b_t^6, a_k\}$ are in the collection $\mathcal{S}$. 
\end{enumerate}

Similar to the \spsc~problem defined in \cite{Chan2014}, in the \spds~problem, we order the elements in $\mathcal{B}$ such that every set in $\mathcal{S}$ contains either two consecutive elements from $\mathcal{B}$, or one element from $\mathcal{A}$ and one element from $\mathcal{B}$, or one element from $\mathcal{A}$ and two consecutive elements from $\mathcal{B}$. In particular, in the below embedding, for every $t$ ($1\leq t \leq m$), the six points  $b_t^1, b_t^2, b_t^3, b_t^4, b_t^5, \text{ and } b_t^6$ are together in the same order. Further, similar to the \spsc~problem, in any instance of the \spds~problem,  every element in the ground set $\mathcal{U}$ is present in exactly two sets in $\mathcal{S}$ and a set in $\mathcal{S}$ contains at most three elements from $\mathcal{U}$.  

%One can encode any instance of the \spds~problem into the instances of the \mdds~problem for classes $\textbf{A1}$-$\textbf{A10}$ by following the similar procedure of obtaining the encoding of the corresponding classes $\textbf{A1}$-$\textbf{A10}$ from the \spsc~problem given in \cite{Chan2014}.
%set covers for the same classes of objects $\textbf{A1}$-$\textbf{A10}$ from the  \spsc~problem given in \cite{Chan2014}. 
%In particular, in \cite{Chan2014}, for a  given instance of the \spsc~problem,  the authors  consider  a point in the plane for every element in ${\cal U}$ and consider a geometric object for every set covering the points corresponding to the elements in that set. For a given instance $(\mathcal{U}, \mathcal{S})$ of the \spds~problem, we also place a point in the plane for every element in $\mathcal{U}$ and for every set in $\mathcal{S}$,  we place a geometric object covering the points corresponding to the elements in that set. In Fig.  \ref{encoding}, we depict the encoding of classes $\textbf{A1}, \textbf{A3},$ and $\textbf{A4}$. 

In the embedding, we consider a point for each element in $\mathcal{U}$, and for each set in $\mathcal{S}$,  we consider a geometric object.  Below, we explain the embedding for each class of objects. 

\begin{description}

\item[- \textbf{(A1):}] We place all the points  in $\mathcal{A}$, in the order $a_1, a_2, \ldots, a_n$, on a segment $\{(x, x-2) \mid x \in [1, 1+\epsilon] \}$ (for a small $\epsilon >0$). Further, place all the points in $\mathcal{B}$ on the segment $\{(x, x+2) \mid x \in [-1, -1+\epsilon] \}$. We note that the sets in $\mathcal{S}$ can be encoded as fat rectangles covering the respective points, as shown in Fig.  \ref{fig-a1}. 

\item[- \textbf{(A2):}] This embedding is similar to the case (A1), except that here each set in $\mathcal{S}$ is encoded as a fat ellipse as opposed to a fat rectangle in (A1).

\item[- \textbf{(A3):}] We place the points in $\mathcal{A}$ on a horizontal line in the order $a_1, a_2, \ldots, a_n$ with the sufficient gap between any two consecutive points as shown in Fig.  \ref{fig-a3}. Recall that each $a_i$ ($i=1,2,\ldots, n$) is contained in exactly two sets in $\mathcal{S}$. If $\{a_i, b_t^1 \}$ is the first set containing $a_i$ (the other case is symmetric), then place $b_t^1$ slightly to the left of $a_i$ and also place $b_t^2$ slightly to the left of  $b_t^1$ such that a vertical strip covers $a_i$ and $b_t^1$, and a horizontal strip covers $b_t^1$ and $b_t^2$. The position of $b_t^3$ (resp. $b_t^6$)  depends on whether the set $\{a_j, b_t^3, b_t^4\}$ (resp. $\{a_k, b_t^6\}$ is either the first or the second set that contains  $a_j$ (resp. $a_k$). Further, $b_t^4$ is vertically below $b_t^3$, and a vertical strip is placed to cover $a_j, b_t^3, \text{ and } b_t^4$. Also, we place a vertical strip to cover $a_k$ and $b_t^6$. The point $b_t^5$ is placed such that a horizontal strip covers the points $b_t^4$ and $b_t^5$ and points $b_t^5$, and a horizontal strip also covers $b_t^6$. 

\item[- \textbf{(A4):}] This is similar to the case (A3), with strips replaced by thin rectangles such that each pair of rectangles either intersect exactly zero or four times. See Fig.  \ref{fig-a4}.

\item
\textbf{(A5):} We place the points in the set $\mathcal{A}$ on the ray $\{(x, -x) : x >0 \}$, in the order $a_1, a_2, \ldots, a_n$,  and place the points in the set $\mathcal{B}$ on the ray $\{(x, x) : x <0 \}$. For each set in $\mathcal{S}$, we place a downward shadow of a segment that covers the corresponding points,  as shown in Fig.  \ref{fig-a5}. 

\item[- \textbf{(A6):}] This embedding is similar to the case (A5). We place all points in $\mathcal{A}$ on a segment $l_1 = \{(z, z) \mid z \in [-1, -1+\epsilon]$, in the order $a_1, a_2, \ldots, a_n$,  and the points in $\mathcal{B}$ are placed on a segment $l_2  = \{(z, 0) \mid z \in [1.5, 1.5+\epsilon] \}$. For any given $(a, a) \in l_1$ and $(b, 0) \in l_2$, the function $f(x) = (x-b)^2[(a+b)x-2a^2] / (b-a)^3$ is tangent to $l_1$ at $x=a$ and tangent to $l_2$ at $x=b$. Thus, the sets of size two in $\mathcal{S}$ can be encoded as cubic polynomials tangent to $l_1$ and $l_2$ at respective points. Further, the sets of size three, $\{a_j, b_t^3, b_t^4 \} \in \mathcal{S}$  can also be encoded as cubic polynomials by considering the cubic polynomial tangent to $l_1$ at $a_j$ and tangent to $l_2$ at $b_t^3$, and slightly shift it upward such that it covers $b_t^4$ also (placing $b_t^3$ and $b_t^4$ sufficiently close).

\item[- \textbf{(A7):}] Place the points in $\mathcal{A}$ and $\mathcal{B}$ on circular arcs $arc_\mathcal{A} = \{(x, y, 0) \mid x^2+y^2=1, x,y \geq 0   \}$ and $arc_\mathcal{B} = \{(0, 0, z) \mid z \in [1-2\epsilon, 1-\epsilon]\}$, respectively.  The sets in $\mathcal{S}$ can be encoded as unit balls in $\mathbb{R}^3$ (see \cite{Chan2014} for full details).

\item[- \textbf{(A8):}] The embedding is similar to (A1). The points in $\mathcal{A}$ are placed on a  segment $l_1 = \{(x, x, 0) \mid x \in (0, 1)\}$ and the points in $\mathcal{B}$ are placed on a  segment $l_2 = \{(0, 3-x, x) \mid x \in (0, 1) \}$.  For any point $p=(x, x, 0) \in l_1$ and any point $q=(0, 3-y, y) \in l_2$, the cube $[-3+y+2x, x] \times [x, 3-y] \times [-3+x+2y, y]$ is  tangent to $l_1$ at $p$ and  tangent to $l_2$ at $q$, and further contains $(0, 1, 0)$. For the sets, $\{a_j, b_t^3, b_t^4\}$ of size three, we can consider the cube that is a tangent to $l_1$ at $a_j$ and a tangent to $l_2$ at $b_t^3$. Further, we can place $b_t^3$ and $b_t^4$ sufficiently close such that the cube covers both points.

\item[- \textbf{(A9):}] This follows from (A7) by using the standard lifting transformation, given in \cite{MarkdeBerg2008}, which maps a point $(x, y, z) \in \mathbb{R}^3$ to a point $(x, y, z, x^2+y^2+z^2) \in \mathbb{R}^4$ and a ball $(x, y, z)$ with $(x-a)^2+(y-b)^2+(z-c)^2 \leq r^2$ to a half-space $(x, y, z, w)$ with $w-2ax-2by+2cz \leq r^2-a^2-b^2-c^2$.

\item[- \textbf{(A10):}] Place the points in $\mathcal{A}$ on the circular arc $arc_\mathcal{A} = \{(\cos t, \sin t): t \in (0, \epsilon)  \}$ and the points in $\mathcal{B}$ on the circular arc $arc_\mathcal{B} = \{\cos t, 2-\sin t): t \in (0, \epsilon)\}$   The sets in $\mathcal{S}$ can be encoded as fat semi-infinite wedges in $\mathbb{R}^2$ (see \cite{Chan2014} for the full details). 

\end{description}

\end{proof}

\begin{figure}[ht!]
\begin{center}
{\subfigure[ ]{\includegraphics[scale=.54]{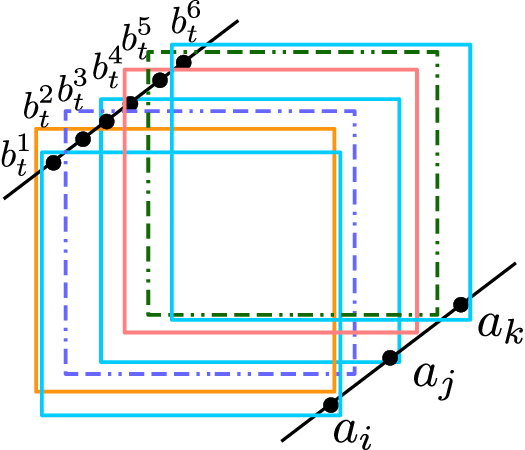}
\label{fig-a1}
}}
{\subfigure[ ]{\includegraphics[scale=.41]{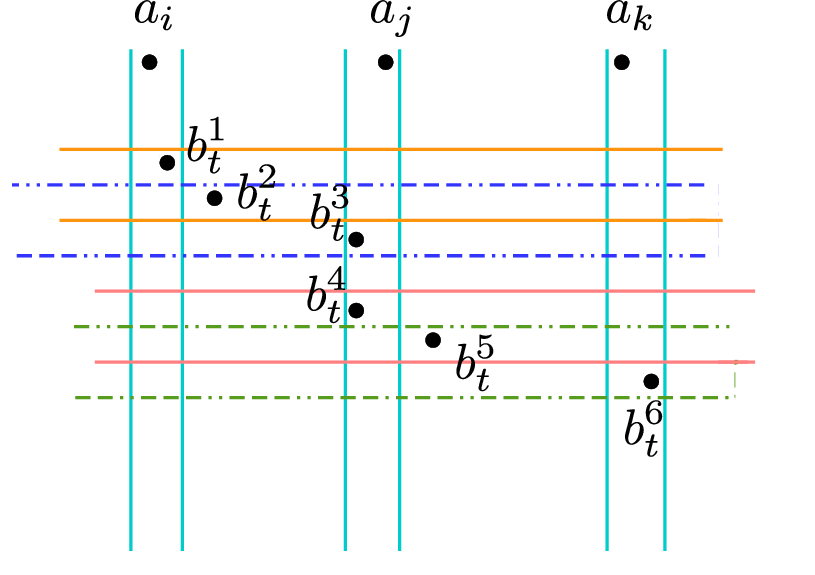}
\label{fig-a3}
}}
{\subfigure[ ]{\includegraphics[scale=.42]{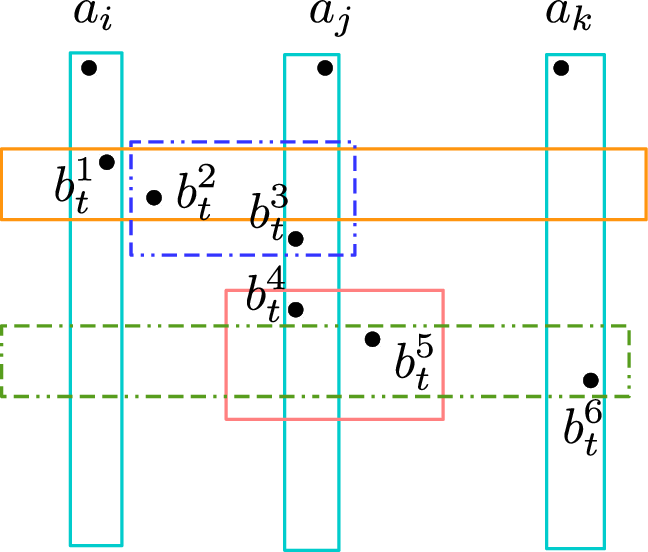}
\label{fig-a4}
}}
{\subfigure[ ]{\includegraphics[scale=.35]{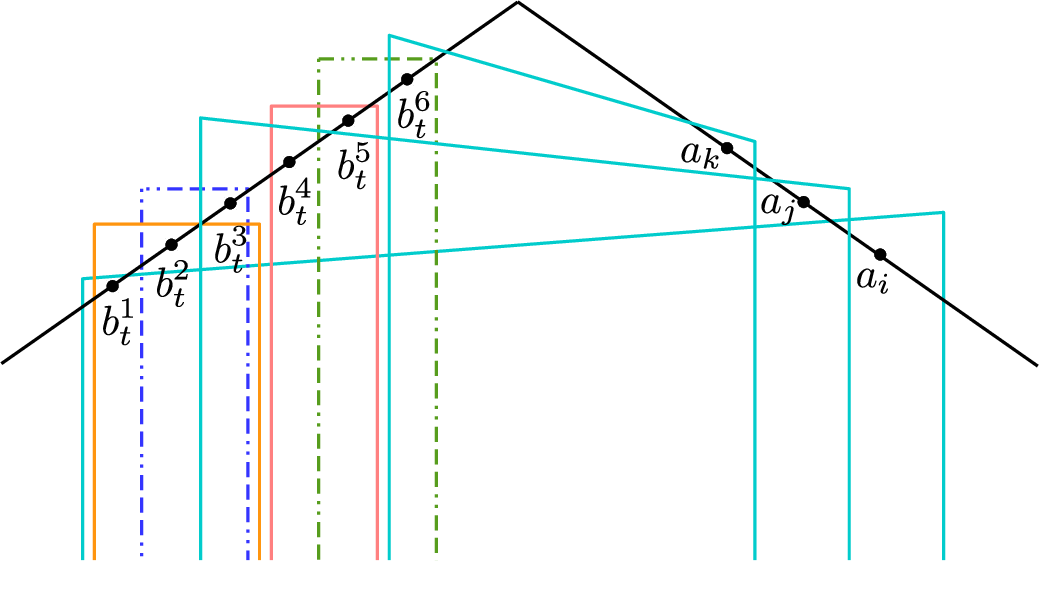}
\label{fig-a5}
}}
\end{center}
\caption{Encoding of \spds~instance into the \mdds~problem instances with various classes of geometric objects. (a) Class $\textbf{A1}$ (b) Class $\textbf{A3}$ (c) Class $\textbf{A4}$ (d) Class $\textbf{A5}$}
\label{encoding}
\end{figure}

We now present some additional \apx-hardness results.

\begin{theorem}
Both \mdis~and \mdds~problems are \apx-hard for the classes of objects $(i)$ fat triangles of similar size and $(ii)$ similar circles. 
\label{more_apx_results}
\end{theorem}
\begin{proof}
The proof is on similar lines to the results of Har-Peled \cite{HarPeled2009}, who showed that the set cover problem  is \apx-hard for rectangles of similar size and similar size circles by giving a reduction from a known \apx-hard problem, the vertex cover problem on cubic graphs \cite{ALIMONTI}.  We first show that the \mdis~problem is \apx-hard for both classes of objects; similar size triangles and similar size circles.

\noindent {\bf The \mdis~problem on similar size fat triangles:}
Let $G= (V, E)$ be a cubic graph. It is known that the independent set problem is \apx-hard on cubic graphs \cite{ALIMONTI}. We now construct an instance of the independent set problem on set systems with range space $\mathcal{X} = (U, \mathcal{S})$. Here, the ground set $U$ contains an element for each edge in the graph $G$, and $\mathcal{S}$ is a collection of $|V|$ subsets of $U$ for each vertex $v$ in $V(G)$, the set $\mathcal{S}$ contains a set $S_v = \{e \mid v \text{ is incident to } e   \text{ and } e \in E(G) \}$. We note that for any independent set of size $t$ for the set system $\mathcal{X} = (U, \mathcal{S})$, there is an independent set for the graph $G$ of the same size $t$. It is known that a graph $G$ with degree 3 is $4$ edge colorable (Vizing's theorem) \cite{BondyMurty}. Thus, one can color the elements in the set $U$ by using four colors such that all the elements in each set $S_v \in \mathcal{S}$ have been assigned a different color. Let $1, 2, 3$, and $4$ be the colors used to color the elements in $U$. Further, for each $i=1, 2, 3, 4$,  let $U_i \subseteq U$ be the set of elements having color $i$. Note that all the sets $U_1, U_2, U_3$, and $U_4$ are pairwise disjoint.

Let $\mathcal{C}$ be the unit radius circle with the center at the origin of the plane. Consider the small circular intervals on the boundary of $\mathcal{C}$ at the intersection with $x$- and $y$- axes. We place the points for the elements in sets $U_1, U_2, U_3$, and $U_4$ at the circular intervals obtained above, one set of points per circular interval. Finally, for each set  $S \in \mathcal{S}$, we consider the convex hull of the points corresponding to the elements in $S$, and the convex hull represents a triangle $T_S$. Here, we note that all such rectangles have similar sizes since these rectangles represent the convex hull of three points such that each point is in different circular intervals defined above.   This gives an encoding of the independent set for $\mathcal{X} = (U, \mathcal{S})$ to the instance of \mdis~problem with similar size triangles. Hence, we conclude that the \mdis~problem is \apx-hard for similar size triangles.  

\noindent {\bf The \mdis~problem on similar circles:} For this case, we slightly perturb the above point-set so that no four points are co-circular. Now, for each set $S \in \mathcal{S}$, we take a circle that passes through the corresponding 3 points. This gives an embedding of the \mdis~problem with similar circles from the independent set problem with set system $\mathcal{X} = (U, \mathcal{S})$. Thus, we conclude that the \mdis~problem is \apx-hard with similar circles.  

Similar reductions of the \mdis~problem for similar size triangles and similar circles lead to the \apx-hardness results of the \mdds~problem for the same classes of objects. However, instead of the maximum independent set problem on cubic graphs, we use the minimum dominating set problem on cubic graphs that are known to be \apx-hard \cite{ALIMONTI}.   
\end{proof}

\section{\texorpdfstring{\boldmath{\np}}~-hardness Results} \label{np_hard_mdis}

In this section, we show that both  \mdis~and \mdds~problems are \np-hard for the following two classes of geometric objects:

\begin{description}
    \item[{\bf B1:}] Unit disks intersecting a horizontal line.  
    \item[{\bf B2:}] Axis-parallel unit squares intersecting a straight  line with slope $-1$. 
\end{description}

For {\bf B1}, the reduction is similar to the reduction of covering points by unit disks where the points and disk centers are constrained to be inside a horizontal strip (the \colb{within strip discrete unit disk cover (WSDUDC)} problem) \cite{Fraser2017}.  On the other hand, for {\bf B2}, the reduction is similar to the reduction of the set cover problem with unit squares where the squares intersect a line with slope $-1$ \cite{Mudgal2015}. For the \mdis~problem, we give a reduction from the known \np-hard problem \textit{maximum independent set}  on planar graphs where the degree of each vertex of the graph is at most 3 (\colb{\misp}~problem) \cite{Garey1977} and for the \mdds~problem  we give a reduction from the \np-hard problem \textit{minimum dominating set}  on planar graphs such that every vertex is of degree at most 3 (\colb{\mdsp}~problem) \cite{Garey1977}.
For the correctness of the reductions, we use the following lemmas.  
\begin{lemma} \label{lemma-mdis-sq-int-diag-np-hard} 
Let $G$ be a graph and $e$ be an edge of $G$, then replacing $e$ by a path with new $2k$ dummy vertices of degree 2 each increases the size of any maximum independent set in $G$ by exactly $k$.
%Let $G$ be a graph and $e$ be an edge of the graph $G$, then adding $2k$ dummy vertices on $e$ increases the size of any maximum independent set in $G$ by exactly $k$.
\end{lemma}

\begin{lemma} \label{lemma-mdds-sq-int-diag-np-hard} 
Let $G$ be a graph and $e$ be an edge of $G$, then replacing $e$ by a path with new $3k$ dummy vertices of degree 2 each increases the size of any minimum dominating set in $G$ by exactly $k$.
%Let $G$ be a graph and $e$ be an edge of the graph $G$, then adding $3k$ dummy vertices on $e$ increases the size of any minimum dominating set in $G$ by exactly $k$.
\end{lemma}

%The reduction is composed with two phases, Phase 1 and Phase 2. In Phase 1, from an instance $G$ of \misp~problem, another instance $G'$ of the same \misp~problem is generated. Next, in Phase 2, from $G'$, an instance ${\tt M_{G'}}$ of the \mdis~problem for {\bf B1} is generated. 

%We omit the constructions and proofs of the hardness results since those can be borrowed from the respective papers (for {\bf B1} \cite{Fraser2017} and for {\bf B2} \cite{Mudgal2015}) as mentioned above. Finally, with the help of  Lemma \ref{lemma-mdis-sq-int-diag-np-hard}, we conclude the following theorem.

We now prove the following theorem.

\begin{theorem}
    Both the \mdis~and \mdds~problems are \np-hard for both $\textbf{B1}$ and $\textbf{B2}$ classes of  objects. 
\end{theorem}

\begin{proof} We first prove that the \mdis~problem is \np-hard for {\bf B1} and {\bf B2} classes of objects. Next, we prove that the \mdds~problem is \np-hard for {\bf B1} and {\bf B2} classes of objects.

\noindent \paragraph{\textbf{The \mdis~problem for {\bf B1}:}}

Here, we use the reduction similar to the {\it WSDUDC} problem \cite{Fraser2017}. We give a reduction from a known \np-hard problem the \misp~problem.  We borrow the constructions and proofs of the hardness result from Fraser and L\'opez-Ortiz \cite{Fraser2017}. For the sake of completeness, we briefly describe the result here.

We make the reduction in two phases, Phase 1 and Phase 2. In Phase 1, from an instance $G$ of the \misp~problem, another instance $G'$ of the same \misp~problem is generated. Next, in Phase 2, from $G'$, an instance ${\tt M_{G'}}$ of the \mdis~problem for {\bf B1} is generated. 

\paragraph{Phase 1 (Constructing $G'$ from $G$):} This phase is identical to \cite{Fraser2017}. In $G$, we add dummy vertices to generate $G'$. The addition is made in the following four steps. Since $G$ is a planar graph, it can be embedded in the plane such that no two vertices of $G$ have the same either $x$- or $y$-coordinates. For a vertex, we say that an edge is incident to it either from the left or right. The edges that are incident to a vertex from exactly one side (either left or right) can be ordered in the $y$-direction.

\begin{description}
    \item[Step 1:] Let $v$ be a degree 3 vertex where all 3 edges incident to $v$ are either from left or from the right. We replace the bottom edge $e$ with either a `<' type edge (if $e$ is incident to $v$ from right) or a `>' type edge (if $e$ is incident to $v$ from left) by adding a new dummy vertex at the corner. See ``\includegraphics[scale=.9]{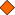}''-shaped vertex in Fig.  \ref{fig-nph-is-2}. Let $G_1$ be the resulting graph generated at the end of this step.
    
    \item[Step 2:] Through each vertex $v$ of $G_1$, draw a vertical line and add a dummy vertex at the intersection point between the vertical line and an edge of $G_1$. See ``\includegraphics[scale=.9]{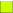}''-shaped vertex in Fig.  \ref{fig-nph-is-3}. Let $G_2$ be the resulting graph generated at the end of this step.

    \item[Step 3:] If the difference between the number of vertices on two consecutive vertical lines differs by more than 1, then add a vertical line between these two consecutive vertical lines. Add a dummy vertex at the intersection point between each newly added vertical line and each edge of $G_2$. See ``\includegraphics[scale=.9]{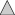}''-shaped vertex in Fig.  \ref{fig-nph-is-4}. Let $G_3$ be the resulting graph generated at the end of this step.

    \item[Step 4:] If the number of dummy vertices added during Steps 1 through 3 to an edge $e$ in $G$ is odd, then consider two consecutive vertical lines $\ell$ and $\ell'$ through two consecutive vertices (maybe dummy vertices) on $e$. We add 2 vertical lines $l_1,l_2$ between $\ell$ and $\ell'$. Add a dummy vertex at the intersection point between each vertical line $l_i$ and each edge of $G_3$. See ``\includegraphics[scale=1]{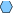}''-shaped vertex in Fig.  \ref{fig-nph-is-5}. Finally, add a dummy vertex immediately to the right of the dummy vertex at the intersection between $e$ and $l_1$. See ``\includegraphics[scale=1]{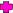}''-shaped vertex in Fig.  \ref{fig-nph-is-6}. Let $G'$ be the resulting graph generated at the end of this step.  
\end{description}

\begin{figure}[ht!]
\begin{center}
{\subfigure[ ]{\includegraphics[width=.46\textwidth]{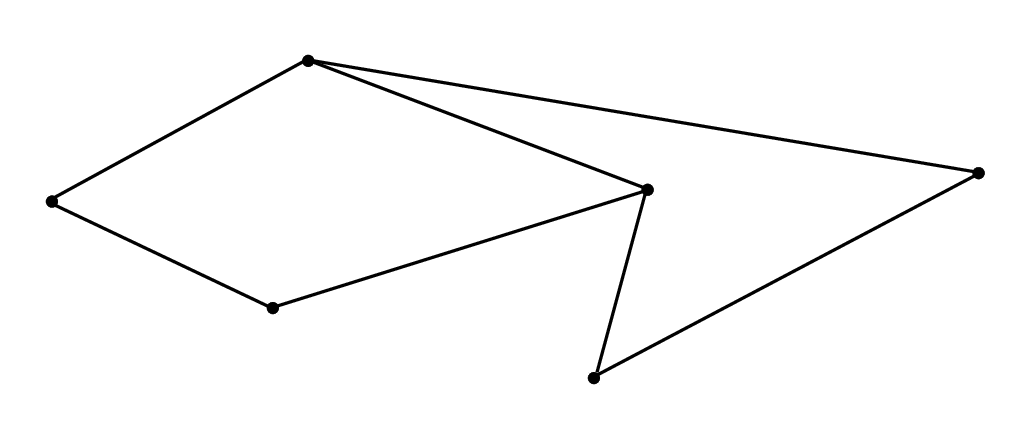}
\label{fig-nph-is-1}
}}
{\subfigure[ ]{\includegraphics[width=.46\textwidth]{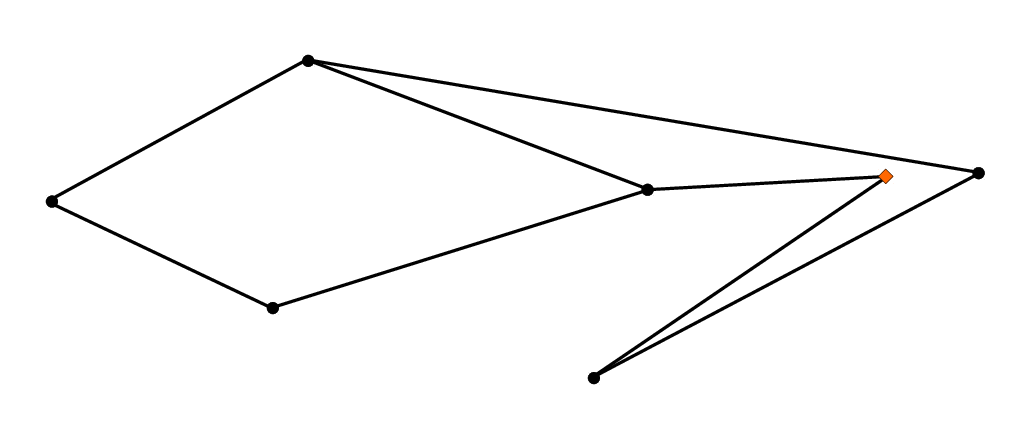}
\label{fig-nph-is-2}
}}
{\subfigure[ ]{\includegraphics[width=.46\textwidth]{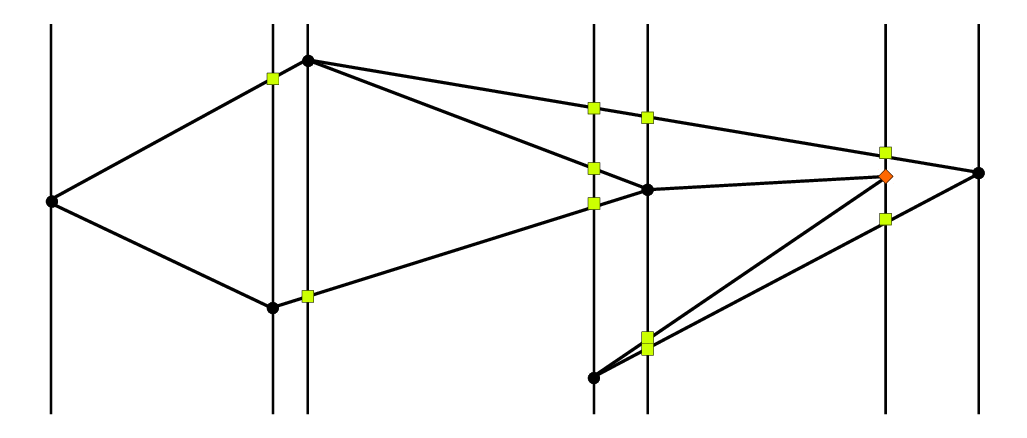}
\label{fig-nph-is-3}
}}
{\subfigure[ ]{\includegraphics[width=.46\textwidth]{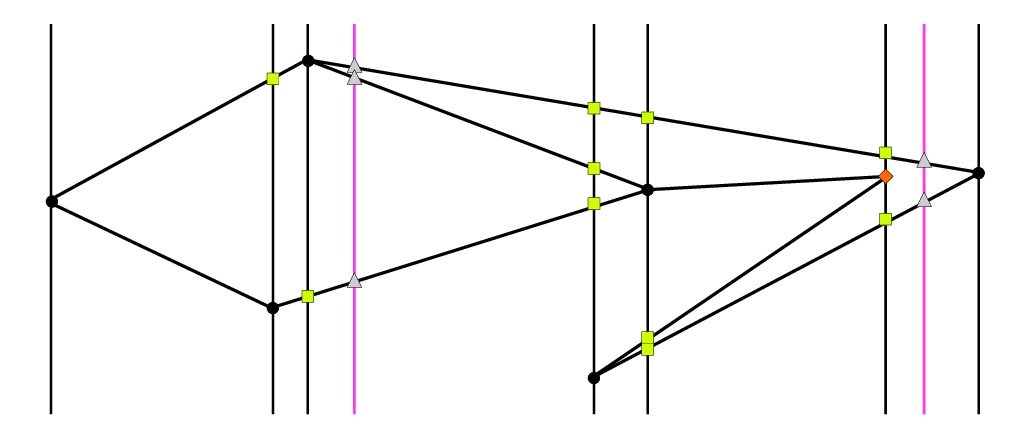}
\label{fig-nph-is-4}
}}
{\subfigure[ ]{\includegraphics[width=.46\textwidth]{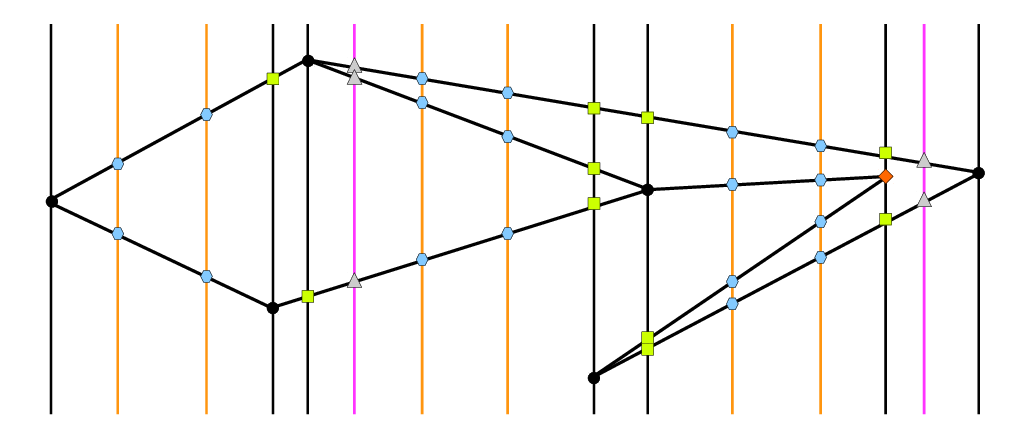}
\label{fig-nph-is-5}
}}
{\subfigure[ ]{\includegraphics[width=.46\textwidth]{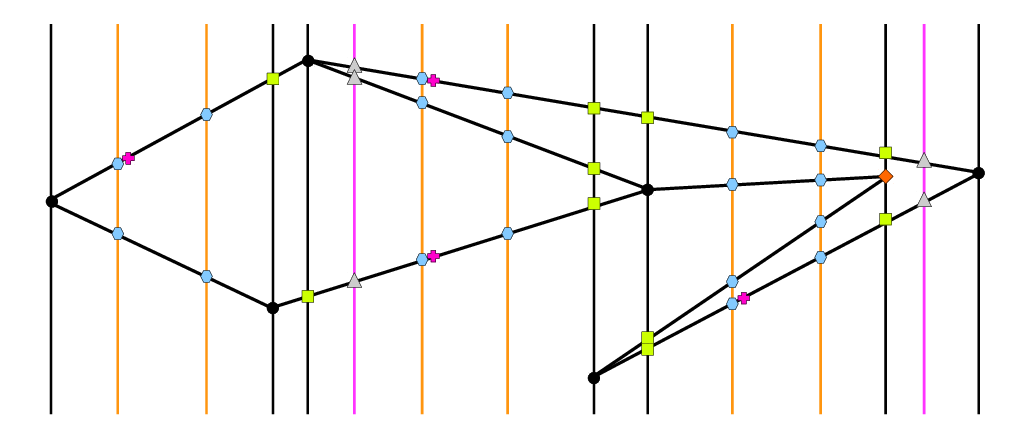}
\label{fig-nph-is-6}
}}
\end{center}
\caption{Different steps of Phase 1 that generates an instance $G'$ of the \misp~problem from an instance $G$ of the same \misp~problem. (a) The given graph, (b) Step 1, (c) Step 2, (d) Step 3, (e)-(f) Step 4.}
\label{fig-nph-is}
\end{figure}

%(b) Step 1, addition of the {``\includegraphics[scale=.9]{figures/fig401.eps}''}-shaped vertices, (c) Step 2, addition of the {``\includegraphics[scale=.9]{figures/fig402.eps}''}-shaped vertices, (d) Step 3, addition of the ``\includegraphics[scale=.9]{figures/fig403.eps}''-shaped vertices, (e)-(f) Step 5, addition of the ``\includegraphics[scale=.9]{figures/fig404.eps}''- and ``\includegraphics[scale=.9]{figures/fig405.eps}''-shaped vertices.

Let $G'$ be the graph returned at the end of Phase 1. Clearly, $G'$ is an instance of the \misp~problem. By applying Lemma \ref{lemma-mdis-sq-int-diag-np-hard}, we say that the \misp~problem for the kind of graph generated in Phase 1 is \np-hard.

\paragraph{Phase 2 (Constructing ${\tt M_{G'}}$ from $G'$)}:

Phase 2 is identical to the construction given in \cite{Fraser2017}. Here, for each vertex $v$ in $G'$, take a unit disk $d_u$, and for each edge $e$ in $G'$, take a point $p_e$ in ${\tt M_{G'}}$. Two vertices $u$ and $v$ are connected by an edge $e$ if and only if their corresponding disks $d_u$ and $d_v$  cover the point $p_e$.

Clearly, ${\tt M_{G'}}$ is an exact embedding of $G'$. Therefore, finding a minimum size independent set of vertices in $G'$ is
equivalent to finding a minimum size independent set of unit disks in ${\tt M_{G'}}$. Hence the \mdis~problem for {\bf B1} is \np-hard.

\noindent \paragraph{\textbf{The \mdis~problem for {\bf B2}:}}

We give a reduction from the \misp~problem. Here, we also make the reduction in two phases. Phase 1 is identical to Phase 1 that we described above for the \mdis~problem for {\bf B1}. We create an instance $G'$ of the \misp~problem from $G$, an instance of the \misp~problem. Clearly, using Lemma \ref{lemma-mdis-sq-int-diag-np-hard}, we can say that the \misp~problem on the type of graph $G'$ generated from the \misp~problem instance $G$ is \np-hard.  

Phase 2 is identical to Phase 2 of the \np-hardness reduction of the Set Cover problem in \cite{Mudgal2015}. We create an instance $M_{G'}$ of the \mdis~problem from $G'$. Here, for each vertex $v$ in $G'$, take a unit square $t_v$, and for each edge $e$ in $G'$, take a point $p_e$. Two vertices $u$ and $v$ are joined by an edge if and only if both $t_u$ and $t_v$ cover the point $p_e$.

Actually, ${\tt M_{G'}}$ is an exact embedding of $G'$. Therefore, finding a minimum size independent set of vertices in $G'$ is
equivalent to finding a minimum size independent set of unit squares in ${\tt M_{G'}}$. Hence, the \mdis~problem for {\bf B2} is \np-hard.

\noindent \paragraph{\textbf{The \mdds~problem for {\bf B1}:}}

Here, we give a reduction from the \mdsp~problem. The reduction is similar to the reduction described for the \mdds~problem for {\bf B1} with a few differences. Here, the reduction is also composed of two phases. In Phase 1, an instance $G'$ of the \mdsp~problem is generated from an instance $G$ of the \mdsp~problem. Next, in Phase 2, an instance $M_{G'}$ of the \mdds~problem with unit squares is generated from $G'$.

To prove that the \mdsp~problem on $G'$ is \np-hard, we apply Lemma \ref{lemma-mdds-sq-int-diag-np-hard}. For this purpose, we must modify only Step 4 of Phase 1 for the \mdds~problem for {\bf B1}. The other Steps remain the same. 

Modification in Step 4:  In order to prove that the \mdsp~problem on $G'$ is \np-hard, we apply Lemma \ref{lemma-mdds-sq-int-diag-np-hard}. Thus in each edge, $e$ in $G$, the dummy vertices added at the end of Phase 1 must be a multiple of 3. To ensure this, we do the following.

Consider an edge $e$ in $G$. Let the number of dummy vertices added on $e$ during Steps 1 through 3 is $d$ that is not a multiple of 3, i.e., $d\neq 3k$ for some integer $k \geq 0$. In this case, consider two consecutive vertical lines $\ell$ and $\ell'$ through two consecutive vertices (maybe dummy vertices) on $e$. We add 6 vertical lines $l_1,l_2,\ldots, l_6$ between $\ell$ and $\ell'$. Add a dummy vertex at the intersection point between each vertical line $l_i$ and each edge of $G_3$. See ``\includegraphics[scale=1]{figures/fig404.eps}''-shaped vertex in Fig.  \ref{fig-nph-ds-1}. Now, two cases can arise.
\begin{itemize}
    \item 
$d=3k+1$ for some integer $k>=0$: As in Step 4 of the \mdis~problem for {\bf B1}, add a dummy vertex immediate to the right of the dummy vertex at the intersection between $e$ and $l_2$. See ``\includegraphics[scale=1]{figures/fig405.eps}''-shaped vertex in Fig.  \ref{fig-nph-ds-2}.

\item $d=3k+2$ for some integer $k>=0$: As in Step 4 of the \mdis~problem for {\bf B1}, add one dummy vertex immediate to the right of the dummy vertex at the intersection between $e$ and $l_2$ and add another dummy vertex immediate to the right of the dummy vertex at the intersection between $e$ and $l_5$. See ``\includegraphics[scale=1]{figures/fig405.eps}''-shaped vertex in Fig.  \ref{fig-nph-ds-2}.
\end{itemize}

\begin{figure}[ht!]
\begin{center}
{\subfigure[ ]{\includegraphics[width=.46\textwidth]{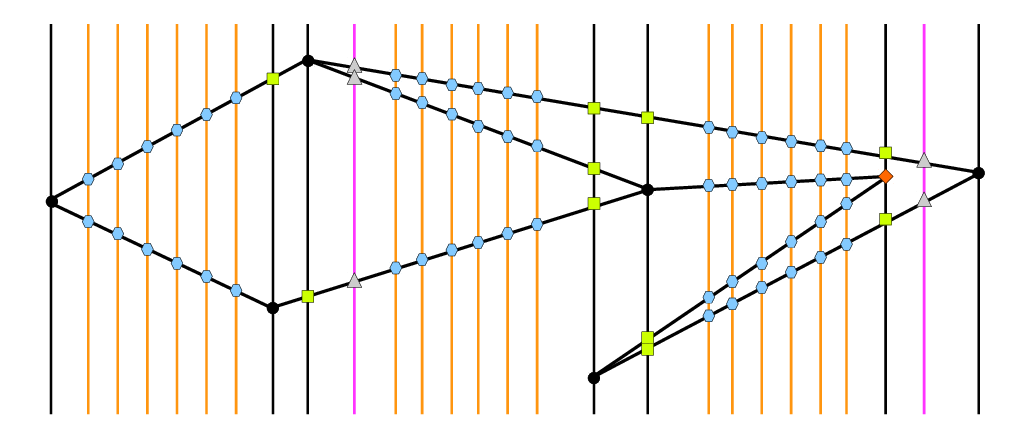}
\label{fig-nph-ds-1}
}}
{\subfigure[ ]{\includegraphics[width=.46\textwidth]{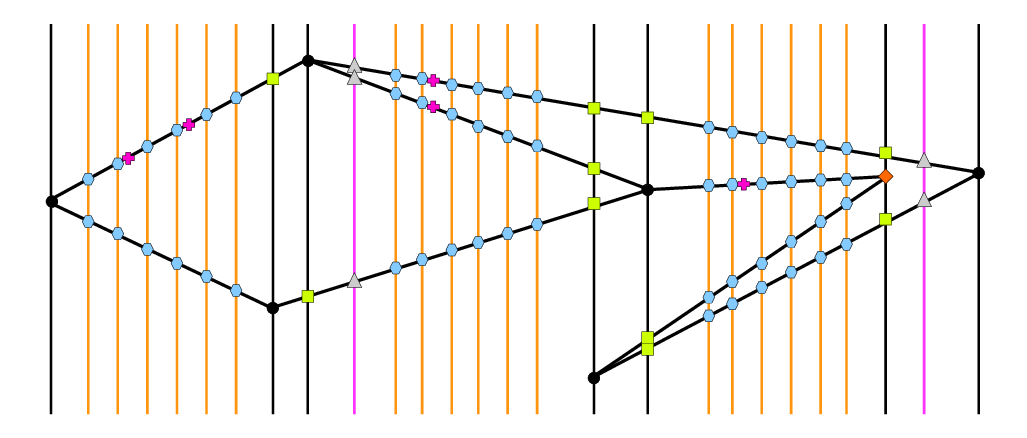}
\label{fig-nph-ds-2}
}}
\end{center}
\caption{(a)-(b) Step 4 of Phase 1 of generating an instance $G'$ of the \mdsp~problem from an instance $G$ of the same \mdsp~problem.}
\label{fig-nph-ds}
\end{figure}

It is easy to observe that finding a minimum size dominating set of vertices in $G'$ is
equivalent to finding a minimum size dominating set of unit disks in ${\tt M_{G'}}$. Hence The \mdds~problem for {\bf B1} is \np-hard.

\noindent \paragraph{\textbf{The \mdds~problem for {\bf B2}:}} For this also we give a reduction from the \mdsp~problem. The reduction consists of two phases. Phase 1 is identical to Phase 1 of the \mdds~problem for {\bf B1} above, and it generates the graph $G'$. Phase 2 is identical with phase 2 of the \mdis~problem for {\bf B2} that generates an instance $M_{G'}$ of the \mdds~problem for {\bf B2}.

Since $M_{G'}$ is an exact embedding of $G'$, it implies that finding a minimum size dominating set of vertices in $G'$ is
equivalent to finding a minimum size dominating set of unit squares in ${\tt M_{G'}}$. Hence The \mdds~problem for {\bf B2} is \np-hard. 

Hence, the theorem is proved. 
\end{proof}

%\begin{theorem}	\mdds~problem is \np-hard for the following classes of objects $\textbf{B1}, \textbf{B2}$, and $\textbf{B3}$. 
%\end{theorem}

%\section{\boldmath{\np}-hardness Results for \boldmath{\mdds}~Problem } \label{np_hard_mdds}


\section{Conclusion}

In this paper, for both \mdis~and \mdds~problems,  we design local search-based \ptas es when the objects are arbitrary radii disks and arbitrary side length axis-parallel squares. These results partially address the question posed by Chan and Har-Peled \cite{Chan2012}  about designing a \ptas~for the \mdis~problem with pseudo-disks.   Further, we show that the \mdds~problem is \apx-hard for various types of geometric objects in $\mathbb{R}^2$ and $\mathbb{R}^3$.   Finally, we prove that both  \mdis~and \mdds~problems are \np-hard for unit disks intersecting a horizontal line and axis-parallel unit squares intersecting a straight line with slope $-1$. A natural open question is the existence of \ptas es for the \mdis~and \mdds~problems with pseudo-disks.

	\bibliographystyle{plain}
%	\newpage
	\bibliography{bib}

\end{document}